\documentclass[11pt,oneside]{article}
\usepackage[margin=1.0in]{geometry}
\usepackage{amsmath}
\usepackage{amsfonts}
\usepackage{amsthm}
\usepackage{amssymb}
\usepackage{enumerate}
\usepackage[T1]{fontenc}
\usepackage{mathrsfs}
\usepackage{mathdots}
\usepackage{graphicx}
\usepackage{subcaption}
\usepackage{mathtools}
\graphicspath{ {./images/} }
\usepackage{xcolor}
\usepackage{hyperref}
\hypersetup{colorlinks=true,linkcolor=red,citecolor=blue,filecolor=magenta,urlcolor=cyan}
\usepackage{subfiles}
\usepackage{pgfplots}
\usepackage{tikz}
\usepackage{bbm}
\usepackage{multirow}
\usepackage{geometry}
\usepackage{array} 
\pgfplotsset{compat=1.18} 
\usepackage{authblk}
\usepackage{algorithm}
\usepackage{algorithmic}
\usepackage{xcolor}
\usepackage{mathtools}

\usepackage[utf8]{inputenc}

\usepackage{natbib}
\usepackage{boldline} 

\newtheorem{theorem}{Theorem}[section]

\newtheorem{lemma}[theorem]{Lemma}

\theoremstyle{definition}

\newtheorem{definition}[theorem]{Definition}

\newcommand{\X}{\mathcal{X}}
\newcommand{\Y}{\mathcal{Y}}

\newcommand{\M}{\mathcal{M}}
\newcommand{\D}{\mathcal{D}}
\newcommand{\Ha}{\mathcal{H}}
\newcommand{\G}{\mathcal{G}}

\newcommand{\E}{\mathop{\mathbb{E}}}

\title{Reconciling Predictive Multiplicity in Practice}
\author[1]{Tina Behzad\thanks{Email: \texttt{tina.behzad@stonybrook.edu}.}}
\author[2]{S\'ilvia Casacuberta\thanks{Email: \texttt{silvia.casacuberta.puig@cs.ox.ac.uk}.}}
        \author[3]{Emily Ruth Diana\thanks{Email: \texttt{ediana@andrew.cmu.edu}.}}
        \author[4]{Alexander Williams Tolbert\thanks{Email: \texttt{alexander.tolbert@emory.edu}.}}
\affil[1]{Stony Brook University}
\affil[2]{University of Oxford}
\affil[3]{CMU Tepper School of Business}
\affil[4]{Emory University}
\date{\today}

\begin{document}

\maketitle

\begin{abstract}
  Many machine learning applications predict individual probabilities, such as the likelihood that a person develops a particular illness. Since these probabilities are unknown, a key question is how to address situations in which different models trained on the same dataset produce varying predictions for certain individuals. This issue is exemplified by the model multiplicity (MM) phenomenon, where a set of comparable models yield inconsistent predictions. Roth, Tolbert, and Weinstein recently introduced a reconciliation procedure, the \textit{Reconcile algorithm}, to address this problem. Given two disagreeing models, the algorithm leverages their disagreement to falsify and improve at least one of the models.
  In this paper, we empirically analyze the Reconcile algorithm using five widely-used fairness datasets: COMPAS, Communities and Crime, Adult, Statlog (German Credit Data), and the ACS Dataset. We examine how Reconcile fits within the model multiplicity literature and compare it to existing MM solutions, demonstrating its effectiveness. We also discuss potential improvements to the Reconcile algorithm theoretically and practically. Finally, we extend the Reconcile algorithm to the setting of causal inference, given that different competing estimators can again disagree on specific causal average treatment effect (CATE) values. We present the first extension of the Reconcile algorithm in causal inference, analyze its theoretical properties, and conduct empirical tests. Our results confirm the practical effectiveness of Reconcile and its applicability across various domains.
\end{abstract}

\newpage

{
  \hypersetup{linkcolor=black}
  \tableofcontents
}

\newpage

\section{Introduction}

In recent years, as algorithms have increasingly been deployed in society to determine life-altering outcomes for individuals, a growing body of research has focused on studying the fairness and robustness properties of these predictors. 
However, many approaches implicitly assume that there is one unique predictor that is adequate for the given task and data.
Recent work has started to pay closer attention to the problem of \emph{model multiplicity} \citep{marx2020predictive}: what should we do when we have competing models that assign conflicting predictions to certain individuals? 
This issue is particularly significant in fairness-sensitive applications, such as loan approvals or predictions of criminal recidivism, where inconsistent individual predictions can lead to arbitrary and potentially discriminatory outcomes. The phenomenon is also known as the \emph{Rashomon effect} \citep{hsu2022rashomon, breiman2003statistical}, which emphasizes the existence of multiple plausible models that explain the data equally well.
\emph{Predictive multiplicity} (PM), a subset of MM, specifically refers to the variability in the predictions of different models within the same accuracy threshold \citep{marx2020predictive}.
Beyond machine learning, model multiplicity is an instance of the broader \emph{reference class problem}, which is a fundamental challenge in probabilistic reasoning that arises has been extensively documented across several kinds of literature from philosophy to the causal inference literature \citep{bolinger2020rational, gardiner2018evidentialism,dahabreh2017heterogeneity,Roth2022-sd}.
The challenges that arise within the MM problem have also recently been explored in the legal domain, where it has been argued that entities should have a legal duty to search for the least discriminatory algorithm within the set of competing models \citep{black2024less,laufer2024fundamental}.

Recently, Roth and Tolbert proposed a new algorithm aimed at reconciling two predictors that agree on the underlying data but disagree on specific individual predictions: they showed that this is resolvable through an efficient reconciliation process that updates the two predictors such that their respective accuracies can only improve, and the two updated models approximately agree on their predictions of individual probabilities almost on the entire domain \citep{Roth2022-sd}.
Throughout this paper, we call this algorithm (formalized in Sections~\ref{sec:notation} and \ref{appendix-sec:reconcilealgo}) \emph{Reconcile}.
While the Reconcile algorithm is theoretically able to resolve the predictive multiplicity problem among predictors that agree on the underlying data,
the original paper by \cite{Roth2022-sd} provides no empirical analysis of the efficacy of the reconciliation process in practice nor further exploration of the various ways in which the Reconcile algorithm can be used on actual predictors and datasets.
Moreover, while the authors mention the connections of the Reconcile algorithm to the model multiplicity problem, they do not explicitly compare the Reconcile approach to the solutions and metrics that have been proposed in the MM literature. In this paper, we investigate the Reconcile algorithm in the context of the model multiplicity problem, and so our empirical studies are centered on previous works on the predictive multiplicity problem.

\paragraph{Our contributions.} We perform an in-depth analysis of the Reconcile algorithm.
For the prediction setting (i.e., the usual PM setting), we provide the first implementation of the Reconcile algorithm and apply it to five of the most popular datasets in fair machine learning:
COMPAS \citep{Larson2016a}, Communities and Crime \citep{misc_communities_and_crime_183}, Adult \citep{misc_adult_2}, Statlog (German Credit Data) \citep{hofmann1994statlog}, and the ACS Dataset from the Folktables package \citep{ding2021retiring}.
For each of the datasets, we generate various model pairs with comparable accuracies but with disagreeing predictions using four different strategies.
We then apply the Reconcile algorithm to each of the pairs and record the evolution of the square losses and the fraction of disagreement between the models throughout the iterations of the algorithm. 
Our results indicate that Reconcile offers a fast-converging solution to reduce disagreement between models to the point that the models become almost identical in all their predictions.
Next, we compare Reconcile to the solutions previously proposed in the model multiplicity literature, clarifying the various conceptual connections and performing an empirical evaluation of the different methods. 
The various solutions proposed in the MM literature can roughly be divided into classes addressing two different goals:
The first and main goal is to determine a way of choosing a \emph{single} model $m$ derived from a class $\M$ of competing models 
in a non-arbitrary way.
A second goal is to quantify the severity of the PM problem in a given class of models $\M$ using different metrics. 
In this case, we do not necessarily wish to select a final model; we only want to quantify the degree of disagreement within the class $\M$. 
As we demonstrate in this paper, Reconcile is an extremely useful algorithm for both use cases.

For the first goal,
given a class of competing models $\M$, we propose a new model aggregation method based on the sequential application of the Reconcile algorithm, which we call \emph{sequential Reconcile}, and compare it with the three solutions arising from the MM literature described in \citep{black2022model}.
For the second goal, we describe different ways in which we can apply the Reconcile algorithm to a class $\M$ of models by repeatedly applying the reconciliation process (since the Reconcile algorithm only takes as input two models at a time).
Through this repeated application of the Reconcile algorithm, we obtain a new class $\M'$ of models that have undergone the reconciliation process.
Next, we apply various metrics proposed in the predictive multiplicity literature to both $\M$ and $\M'$, and we empirically demonstrate how Reconcile drastically reduces the amount of disagreement within the class of models.

We further study the benefits of using sequential Reconcile as model aggregation method.
We extend the original analysis of Reconcile and demonstrate both theoretically and empirically that, as a model aggregator method, Reconcile exhibits strong robustness and fairness guarantees, neither of which necessarily hold for other popular model class aggregation methods such as mean aggregation.
Specifically, on the robustness front, we show that sequential Reconcile is robust to the presence of outlier models in the class $\M$ with lower accuracy, and as long as one model in the class has good accuracy, the final ensembled model will to.
On the fairness front, we demonstrate how sequential Reconcile ensures that the accuracy of models over minority subgroups is preserved under ensembling.
That is, if one model in the class has accurately learnt the predictions for a (large enough) minority group and we reconcile it with predictors that are accurate overall but not over this minority group, then the reconciliation process ensures that the classifiers in the model class all increase their accuracy over the minority group.
This is a relevant and desirable property in societal applications: even if many predictors in the class $\M$ are highly inaccurate over a minority group, it is enough to have a predictor in $\M$ that is accurate for the minority group in order to ensure that the final ensembled predictor has good accuracy over the minority group.

After demonstrating the effectiveness of the Reconcile algorithm in practice in the prediction setting, we provide the first extension of Reconcile to causal inference. 
We analyze the algorithm theoretically, showing guarantees similar to those in the prediction setting, and then evaluate it empirically on the popular Twins dataset \citep{almond2005costs, guo2020survey} and the National Study dataset \citep{nosek2015promoting} in causal inference.
Our results demonstrate that this extension of Reconcile efficiently reconciles treatment effect values between pairs of causal estimators that initially present significant disagreement.
Overall, our paper thoroughly demonstrates and extends the efficacy of the Reconcile algorithm in many different settings.

\section{Related Work}

\paragraph{Reference Class Problem \& Predictive
Multiplicity in Machine Learning.}
\emph{Predictive multiplicity} (PM), a subset of MM, specifically refers to the variability in the predictions of different models within the same accuracy threshold \citep{marx2020predictive}. According to \citep{black2022model}, MM can be categorized into \emph{procedural multiplicity}, where models differ in their internal structures, and PM, where models diverge in their output predictions. PM poses significant challenges in ensuring fairness and consistency, as it complicates the selection of a single model for deployment and raises concerns about the arbitrary nature of individual predictions.
To address the challenge of determining final outcomes amidst conflicting predictions, much of the predictive multiplicity literature emphasizes \emph{ensembling methods}. These methods aggregate different predictions using predefined strategies, such as taking the mode \citep{black2021selective}.
This is a well-known approach for reducing predictive variance \citep{lincoln1989synergy}.
However, since MM arises when models in a class $\mathcal{M}$ share similar accuracies, it is unclear which criteria to prioritize when selecting a single model from $\mathcal{M}$, aside from avoiding arbitrariness. 
Consequently, practitioners often refrain from returning a specific aggregated model derived from $\mathcal{M}$ and instead aim to quantify the extent of prediction variability \citep{hsu2022rashomon}.
Other recent works argue for the need to add an axis of arbitrariness when deploying models, besides reporting accuracy and fairness metrics \citep{long2024individual}.
Predictive multiplicity has also been explored in connection to the problem of target choice \citep{watson2023multi} and to the dataset multiplicity problem \citep{meyer2023dataset}.

However, MM can present opportunities beyond its challenges \citep{black2022model}. When no single model emerges as the most accurate, the existence of a diverse set of similarly accurate models allows practitioners to incorporate additional criteria, such as fairness and robustness, into the model selection process without compromising accuracy \citep{black2022model}. This flexibility can help mitigate the issue of \emph{algorithmic monoculture} \citep{kleinberg2021algorithmic}, where reliance on a single algorithmic approach can lead to systemic vulnerabilities.
An important line of work has introduced several ways of visualizing and comparing the different predictors within a set of good models \citep{dong2019variable, fisher2019all, donnelly2023rashomon, zhong2024exploring}.
Given these considerations, a significant portion of the MM literature is dedicated to \emph{quantifying} the severity of predictive multiplicity within a model class $\mathcal{M}$. By measuring metrics such as \emph{ambiguity} and \emph{discrepancy} \citep{marx2020predictive, watson2023predictive}, researchers aim to characterize multiplicity based on model outputs rather than their internal parameters. These metrics are related to \emph{predictive churn}, which assesses how updated models differ from their predecessors \citep{watson2023predictive}. Often, the model class $\mathcal{M}$ is defined as a \emph{Rashomon set}: given a hypothesis space $\mathcal{H}$ and a parameter $\epsilon \geq 0$, the Rashomon set $\mathcal{R}(\mathcal{H}, \epsilon)$ comprises all models in $\mathcal{H}$ that incur a loss below $\epsilon$. Sometimes, the Rashomon set is defined relative to a fixed baseline model \citep{watson2024predictive}. For example, \cite{hsu2022rashomon} introduce \emph{Rashomon capacity}, a metric for quantifying predictive multiplicity within a Rashomon set. Existing approaches that consider Rashomon sets restrict their analysis to models within the class $\mathcal{H}$. In contrast, the Reconcile algorithm is not confined to a specific hypothesis class.
Within this enlarged class, Reconcile is able to obtain accurate and consistent predictions.

\paragraph{The Reference Class Problem in Causal Inference.}
While model multiplicity was originally introduced in the context of prediction, the model multiplicity problem naturally extends to the setting of causal inference, where we are interested in estimating heterogeneous treatment effects across subgroups.
Indeed, many algorithms in causal inference rely on supervised learning and regression methods from machine learning \citep{kunzel2019metalearners}.
Similar to the prediction setting, different competing causal treatment effect estimators can disagree on their predictions among particular subgroups.
These differences in treatment effects are common in many policy intervention settings, such as in healthcare and epidemiology \citep{dahabreh2017heterogeneity}.
Hence this difference in treatment effects has implications for designing, administering, and evaluating treatments that have optimal outcomes across diverse population segments \citep{Kent2012-jq,steyerberg2005equally}. 

\paragraph{The Reconcile Algorithm.}
Building on Philip Dawid's insights \citep{dawid2017individual}, Roth, Tolbert, and Weinstein recently introduced the \emph{Reconcile algorithm}: they showed that, while individual probabilities are inherently unknowable, they are \emph{falsifiable} \citep{Roth2022-sd}. 
Namely, given two different models $f_1$, $f_2$ that are learned on the same data, we can perform an efficient reconciliation procedure on $(f_1, f_2)$:
If $f_1, f_2$ already agree in their predictions, then we are done.
Otherwise, we can efficiently define a group (which identifies a subset of the individuals) that ``witnesses'' this disagreement.
Then, we can use this group to update either $f_1$ or $f_2$, so that the updated model now makes correct predictions on average with respect to this fixed group. 
We can iteratively apply this procedure until the two models no longer significantly disagree.
\cite{Roth2022-sd} show that this process converges and that each of these updates can only improve the accuracy of the models.
In a concurrent work, \cite{du2024reconciling} extend the Reconcile algorithm for downstream decision making, so that the predictors agree not only on individual predictions but also on the best-response actions for each individual in the downstream decision-making task.

\section{Notation and Preliminaries}\label{sec:notation}
\subsection{Prediction Setting}
We work over a domain $\X \times \Y$, where $\X$ denotes a finite feature domain and $\Y$ denotes the set of labels. 
Following \citep{Roth2022-sd}, we restrict $\Y$ to the set of binary labels $\Y = \{0,1\}$.
In our evaluations, we consider both the regression and classification settings.
Next, we formalize the framework of individual probabilities.
We let $f^*$ denote the ground-truth probabilities for each $x \in \X$.
That is, $f^*: \X \rightarrow [0,1]$ denotes the function that, for each $x \in \X$, returns the value $f^*(x) = \Pr_{(x, y) \sim \D}[y=1 | x]$,
where $\D$ is some probability distribution over $\X \times \Y$.
Given a distribution \( \D \) on $\X \times \Y$ and a group \( g : \X \rightarrow \{0,1\} \), we say that $g$ has probability mass \( \mu(g) \), where $\mu(g) \triangleq \Pr_{(x,y) \sim \D} [g(x) = 1]$.
Our goal is to build a \emph{model} $f: \X \rightarrow [0,1]$ of $f^*$.
Note that we observe the realized Boolean labels $y$, but we never have access to the \emph{individual probabilities} $f^*$.
Still, we can evaluate how ``good'' a proposed model $f$ is with respect to $f^*$ by computing its \emph{Brier score} over the observed labels $B(f, \D) \triangleq \E_{(x, y) \sim \D}[(f(x)-y)^2]$.
We can estimate the Brier score $f$ given enough samples from the distribution. 
In our empirical analysis, we will take $\D$ to be the uniform distribution over the given dataset $D = \{(x_i, y_i)\}$.
Note that the Brier score is minimized at $f = f^*$: if we have two models $f_1$ and $f_2$ such that $B(f_1, \D) < B(f_2, \D)$, then it follows that $f_2$ cannot be equal to $f^*$.
It is through this characterization with the Brier score that \citep{Roth2022-sd} are able to contest two disagreeing models $f_1$ and $f_2$ and falsify one.

\paragraph{Set-up for the Reconcile algorithm.} 
Following the notation in \citep{Roth2022-sd}, we introduce additional definitions for the Reconcile algorithm.
We say that two models $f_1, f_2$ have an \( \epsilon \)-disagreement on \( x \in \X \) if $|f_1(x) - f_2(x)| > \epsilon$.
Having fixed $f_1$ and $f_2$, we denote by $U_{\epsilon}(f_1, f_2)$ the set of points $x \in \X$ on which $f_1$ and $f_2$ have an $\epsilon$-disagreement.
Further, we write $U_{\epsilon}(f_1, f_2) = U^{>}_{\epsilon}(f_1, f_2) \cup U_{\epsilon}^{<}(f_1, f_2)$, where $U^{>}_{\epsilon}(f_1, f_2)$ denotes the points $x \in U_{\epsilon}(f_1, f_2)$ such that $f_1(x) > f_2(x)$, and similarly for $U^{<}_{\epsilon}(f_1, f_2)$.
We are interested in disagreement regions that have at least $\alpha$ probability mass, where $\alpha>0$ is a user-specified parameter.
In the literature on multigroup fairness \citep{hebert2018multicalibration}, the key idea on how to test for accuracy or calibration for a model $f$ (with respect to the underlying $f^*$) is as follows: we cannot know any particular value $f^*(x)$ for a given $x \in \X$, but we can instead consider a group $g: \X \rightarrow \{0,1\}$ 
and measure the average of the realized boolean outcomes $y$ over $g$.
Then, we can compare the average value given by the $y$ labels over $g$ versus the average value predicted by our proposed model $f$. 
Reversing this argument, we can demonstrate a violation of accuracy/calibration on $f$ by finding a group $g$ such that $f$ does \emph{not} satisfy the required property on average over $g$.

In \citep{Roth2022-sd}, the desired property on which we want to test for possible violations is $\alpha$-approximate group conditional mean consistency (see Definition 3.3 in \citep{Roth2022-sd}, which we reproduce in Definition~\ref{definition:meanconsistency} in Section~\ref{appendix-sec:reconcilealgo}).
The key idea behind their algorithm is the following: if two models $f_1, f_2$ have at least an $\epsilon$-disagreement, then at least one of $U^{>}_{\epsilon}(f_1, f_2)$ or $U_{\epsilon}^{<}(f_1, f_2)$ witnesses a large enough violation of group conditional mean consistency. 
We can then use the group witnessing the violation to update the corresponding model, in a way that ensures that we make progress measured in terms of the Brier score. 
In \citep[Lemma 3.2]{Roth2022-sd}, they show that given any model $f_t$ and group $g$ witnessing a violation of $\alpha$-approximate group conditional mean consistency on $f_t$, we can efficiently produce a model $f_{t+1}$ that improves its Brier score by at least $\alpha \epsilon^2$.
We defer the formal description of Reconcile to the appendix; Figure~\ref{fig:reconcile} summarizes the algorithm.
Let $T_1$ denote the total number of updates performed to $f_1$, and likewise $T_2$ for $f_2$.
\cite{Roth2022-sd} prove the following main theorem, which we analyze in our empirical results:
\begin{theorem}[\citep{Roth2022-sd}]\label{thm:mainthm}
    Given any pair of models $f_1, f_2: \X \rightarrow [0,1]$, distribution $\D$ on $\X \times \Y$, and parameters $\alpha, \epsilon>0$, the two models $(f^{T_1}_1, f^{T_2}_2)$ returned by Reconcile when run on $(f_1, f_2, \D, \alpha, \epsilon)$, where $T = T_1+T_2$, satisfy the following properties:

    (1) The algorithm terminates quickly: $T \leq (B(f_1, \D) + B(f_2, \D)) \cdot \frac{16}{\alpha \epsilon^2}$.

    (2) Reconcile improves the Brier scores of the models: $B(f_1^{T_1}, \D) \leq B(f_1, \D) - T_1 \cdot \frac{\alpha \epsilon^2}{16}$ and $B(f_2^{T_2}, \D) \leq B(f_2, \D) - T_2 \cdot \frac{\alpha \epsilon^2}{16}$.

    (3) Reconcile reduces disagreement between models: $\mu(U_{\epsilon}(f_1^{T_1}, f_2^{T_2})) < \alpha$.
\end{theorem}

\cite{Roth2022-sd} provide the generalization bounds of Theorem~\ref{thm:mainthm}: we can run the Reconcile on a set of $n$ samples drawn from $\D$, and they show that with high probability over the sample of $D$, the guarantees of the main theorem translate over to $\D$, with error parameters that go to zero with the size of the data sample.

\begin{figure}
\centering
\includegraphics[width=0.8\columnwidth]{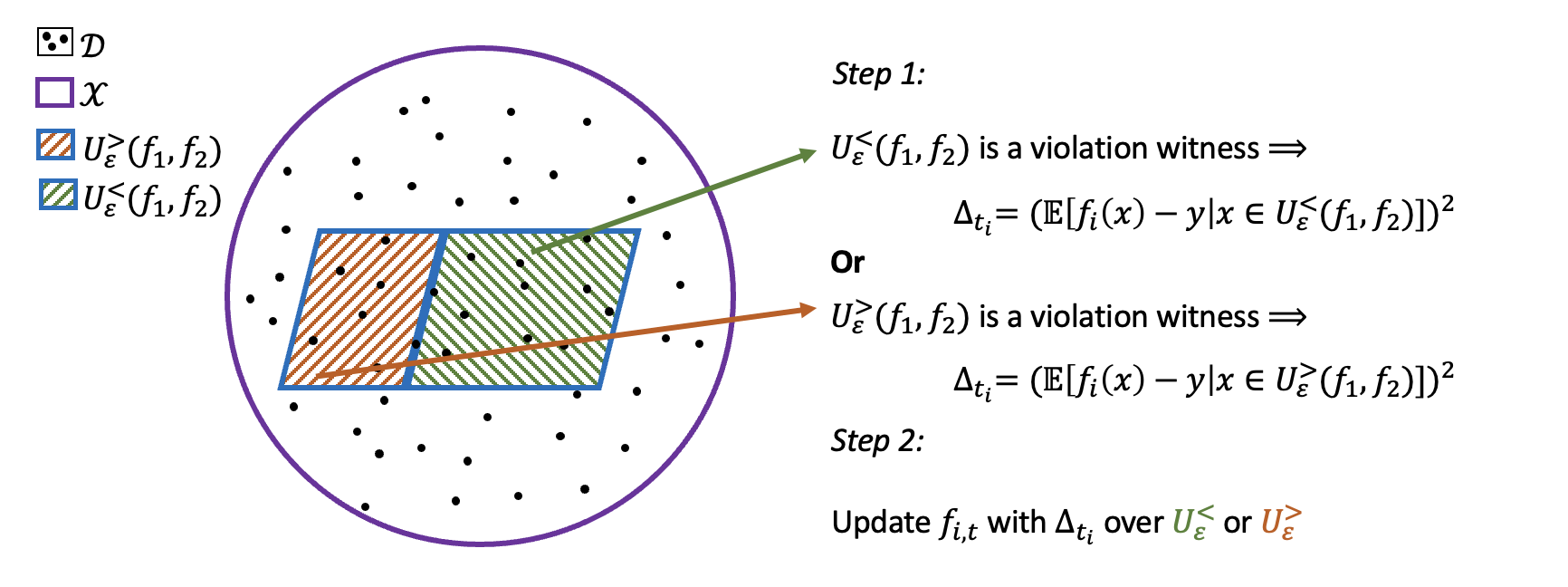}
\caption{Diagram of the Reconcile algorithm.}\label{fig:reconcile}
\end{figure}

\subsection{Causal Inference (CI) Setting}
We follow the standard notation and terminology from the causal inference (CI) literature \citep{ding2024first}.
We receive a dataset $\D = \{(X_i, Y_i, T_i)\}_{i=1}^N$ from a randomized experiments or observational study, where $X_i \in \X$ denotes the vector of covariates of the $i$-th individual, $T_i \in \{0,1\}$ represents the binary treatment assignment, and $Y_i \in \{0,1\}$ denote the outcome.
Ideally, we would like to compute the individual causal effect $\tau_i = Y_i(1)-Y_i(0)$.
However, given that only one of the potential outcomes is observed for each individual, we need to take averages.
In this paper, we are interested in the \emph{conditional average treatment effect} (CATE), which is defined as $\tau(x) = \E[Y(1)-Y(0) \mid X=x]$.
It represents the expected treatment effect for individuals with covariates $X=x$.

As usual in the CI literature, we impose the assumptions of (1) unconfoundedness: $\{Y(1), Y(0)\} \perp T \mid X$, (2) consistency and SUTVA: $Y_i = Y_i(T_i)$, and (3) overlap: $0 < \Pr(T=1 \mid X) <1$.  
We also assume that the outcomes $Y(1)$ and $Y(0)$ are bounded \citep{kern2024multi, kunzel2019metalearners}.
In the setting of CI, we adapt the Brier score to CATE estimators, which is defined as $B(\hat{\tau}, \mathcal{D}) = \mathbb{E}_{(x,y,t) \sim \mathcal{D}} \left[ \left( \hat{\tau}(x) - \tau(x) \right)^2 \right]$.

\section{Empirical Evaluation of Reconcile}\label{sec:experiments}
We now present an extensive empirical analysis of the Reconcile algorithm, and we try demonstrating how it fits within the broader PM literature.\footnote{The code is available at \url{https://github.com/tina-behzad/Reconciliation-Project}.}
A diagram of all of our experiments can be found in Appendix \ref{appendix: experiments}.
In the context of the MM literature, we show how we can use Reconcile both to (1) choose a model from a class $\M$ of competing models and to (2) reduce the severity of the PM problem within $\M$.

We use five datasets in our experiments. These datasets are commonly used in the MM literature and all include high-stake social predictions in which predictive multiplicity can have significant consequences. Specifically, we use the UCI Adult dataset \citep{misc_adult_2}, the COMPAS dataset \citep{Larson2016a}, the Communities and Crime dataset \citep{misc_communities_and_crime_183}, the Statlog German Credit Datset \citep{statlog_(german_credit_data)_144}, and the American Community Survey (ACS) dataset from \citep{ding2021retiring}. A full description of the datasets and the corresponding prediction tasks can be found in Appendix \ref{appendix:datasets}. For ACS, we perform three different prediction tasks: whether an individual’s income is above $\$50,000$ (referred to as Folk\_income in the text), whether an individual had the same residential address one year ago (Folk\_Mobility), and whether an individual has a commute to work that is longer than 20 minutes (Folk\_Travel). Having these datasets, we conduct an analysis of the Reconcile algorithm across both classification and regression tasks. We acknowledge that our evaluation is based exclusively on tabular datasets. However, and importantly, since Reconcile operates on predictions rather than the underlying data, its effectiveness is independent of the data type (e.g., image, text, etc.).

\subsection{Building Pairs of Models to Reconcile} \label{sec:building models}
Prior to the execution of Reconcile, we need to generate two models $f_1, f_2$ with comparable accuracy but significant disagreement in their point-wise predictions. The criteria for these conditions are:

\textbf{ Comparable accuracy}: Brier scores for classification tasks and Mean Squared Errors (MSE) for regression tasks\footnote{We note that the formal definition for the Brier score and MSE are the same. Following the general applications of the two, we refer to it as Brier score for classification tasks and MSE for regression tasks.} differing by no more than 0.05.

\textbf{Significant disagreement}: We define significant disagreement as \( \mu(U_{\epsilon}(f_1, f_2)) \geq \alpha \) where \( \epsilon = 0.2 \) and \( \alpha = 0.05 \). In all the predictive experiments we use the same values for $\alpha$ and $\epsilon$ when running Reconcile.

We employ three different strategies to generate $f_1, f_2$:
\begin{enumerate}
    \item Different model types for $f_1, f_2$ that are trained on the same data (e.g., logistic regression for $f_1$ and K-Nearest Neighbors for $f_2$). 
    \item Same model types for both $f_1$ and $f_2$ that are trained on different (disjoint) subsets of the data, where both subsets include all features.
    \item Same model types for both $f_1$ and $f_2$ that are trained on different (disjoint) subsets of features.
\end{enumerate}

Previous research on MM has explored methods for identifying the complete Rashomon set in logistic regression models and decision trees \citep{dong2019variable, xin2022exploring}.

\subsection{Reconcile Algorithm in Action} \label{sec:experiments Reconcile results}
Using the datasets described in the previous section we train the different models, allocating 60\% of each dataset to create the training set. The models were trained using the four approaches outlined in Section \ref{sec:building models}.
We use the decision tree regressor, $K$-neighbors regressor, and linear regression classes for regression tasks on the Communities dataset and the decision tree classifier, the $K$-nearest neighbor classifier, and logistic regression classes from the scikit-learn package \citep{scikit-learn} for classification tasks on all other datasets.\footnote{Reconcile is inherently model-agnostic, operating directly on predictions rather than the underlying models. As a result, its effectiveness remains consistent regardless of whether sophisticated deep-learning models or simple linear regression are used.}
After finding initial models $f_1$ and $f_2$, 
we run the Reconcile algorithm, using the validation set (which includes 20\% of the data) and analyze the result of Reconcile on the test set (the remaining 20\% of each dataset). 

To find $U_{\epsilon}$ in classification tasks with binary outcomes $0$ and $1$, we use the predicted probability of belonging to class $y=1$. We repeat each experiment 100 times. In some cases, we were not able to obtain models with a comparable accuracy that had a significant disagreement in their point-wise predictions. In total, we were able to run Reconcile $4,517$ times with different models and datasets.

We compare the empirical results with the three bounds provided in Theorem~\ref{thm:mainthm}.
In all scenarios, the algorithm consistently terminated within only $1$ to $7$ rounds.
This is significantly lower than the theoretical bounds specified in Theorem 3.1, which predicted much higher values (minimum of 344, maximum of $5,256$, and an average of $2,817.5$). This discrepancy suggests that this theoretical bound is extremely loose. It also indicates how quick reconcile converges, making it efficient enough to use in any scenario.
We further compare accuracy levels before and after Reconcile. 
The final scores were consistently lower than or equal to the bounds established by Theorem 3.1. A comparison of Brier/MSE scores for both models before and after applying Reconcile across all experiments shows that Reconcile leads to lower minimum, median, and maximum scores (indicating improved performance) across all datasets. These results are visualized in Figure \ref{fig:brier_boxplot} in Appendix \ref{appendix:results}.

Finally, we compare disagreement levels before and after running Reconcile.  Consistent with Theorem 3.1, the final disagreement level is always lower than the initial disagreement level. Furthermore, in nearly all cases examined, the level of disagreement is reduced to zero following reconciliation (see Figure \ref{fig:disagreement_boxplot}). We can see that Reconcile drops the disagreement substantially across all datasets. When disagreements fall to zero, the models are now in complete agreement on all predictions.  
In summary, in all cases Reconcile converges quickly, increases the accuracy of both models, and reduces disagreement levels between the models. 

\begin{figure}
\centering
\includegraphics[width=1\columnwidth]{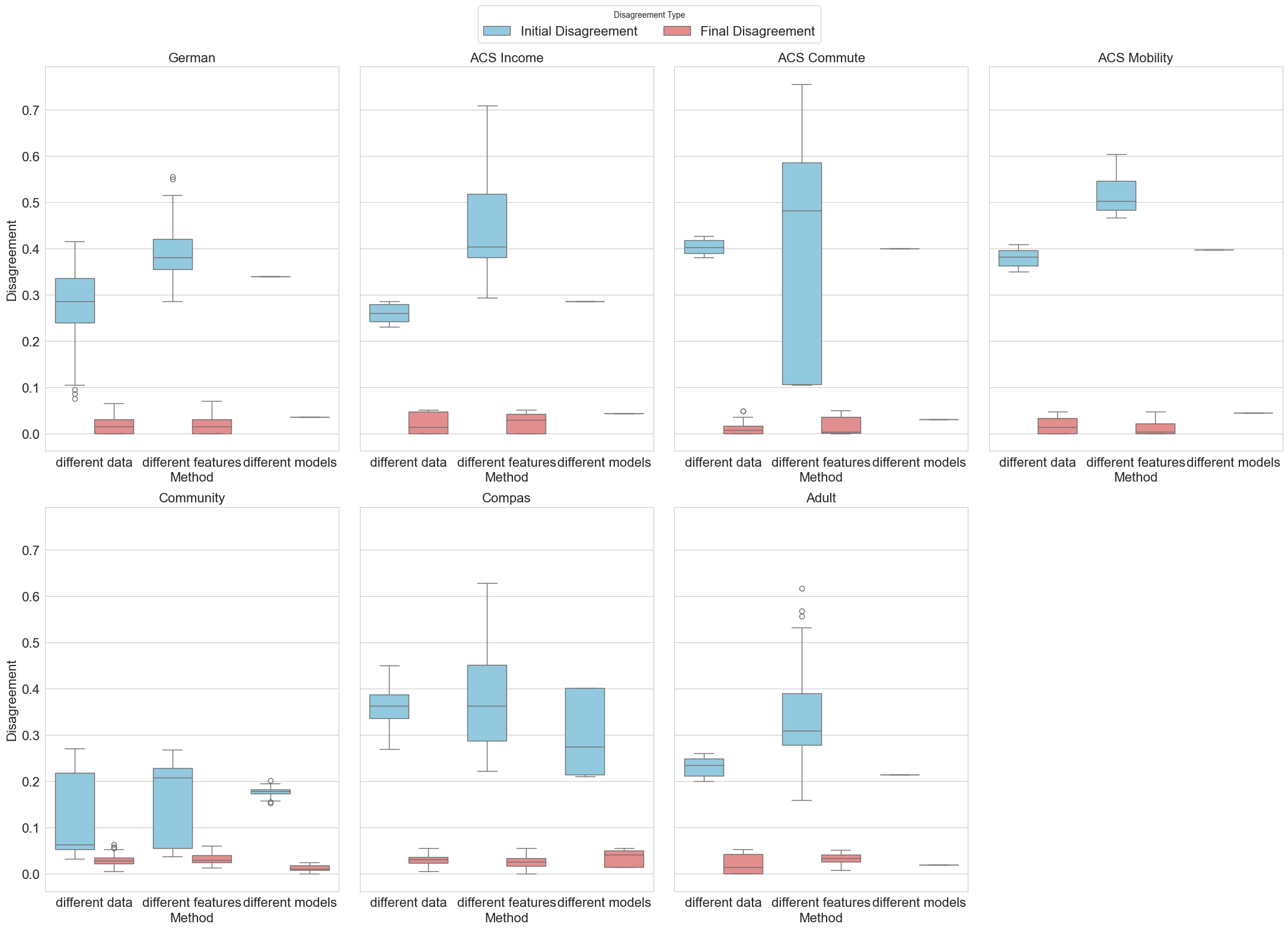} 
\caption{ Disagreement levels between $f_1$ and $f_2$ before and after Reconcile. Model construction follows methodologies described in Section \ref{sec:building models}. Model specifics can be found in Section \ref{sec:experiments Reconcile results}.}
\label{fig:disagreement_boxplot}
\end{figure}

\subsection{Reconcile for Model Class Aggregation} \label{sec:experiments Reconcile and MM}

We now turn our attention back to the MM literature and the question of how to determine final outcomes when faced with disagreeing predictions. In \citep{black2022model}, the authors advocate for a non-arbitrary method to choose between high-accuracy models. This involves creating a ``meta-rule'', which is a set of specified criteria that models must meet to be considered for selection. 
For instance, a meta-rule might dictate that a model must maintain a certain level of accuracy or exhibit equitable true positive rates across demographic groups.
They suggest three \emph{aggregation techniques} as a way to make decisions from a set of models $\M$:
\begin{enumerate}
    \item \textbf{Mode Aggregation}: The mode predictor $f$ aggregates models from the set $\M$ by outputting the majority vote over the models $m \in \M$ for each $x$.
    \item \textbf{Randomized Predictions}: With this approach, the decision maker applies the classifier $m^{\mathrm{rand}}$, which, for each instance $x$, chooses a model $m$ at random from the set $\M$ and outputs $m(x)$. 
    \item \textbf{Random Model Selection}: In random model selection, a model $m$ is sampled at random from the set $\M$ and is used consistently for all $x \in \X$.
\end{enumerate}

These techniques aim to reduce arbitrariness in model selection. In this section, we propose a fourth, novel non-arbitrary method for choosing a model between the models in $\M$ using the Reconcile algorithm.
As we discuss at the end of this subsection, the reconciliation process as a model aggregation technique presents far more benefits than that of non-arbitrariness. 

4. \textbf{Sequential Reconcile}: We construct a model $f$ from the class $\M$ as follows.
\textbf{(a)} Randomly select two models from the set $\M$ and run Reconcile on these models.
\textbf{(b)} Randomly choose one of the two reconciled models (or the one with the lower Brier score)\footnote{Recall that after running the Reconcile with parameter $\alpha$, the two output models do not present significant disagreement. Thus, for small values of $\alpha$, the two output models are essentially equivalent (and, in some cases, exactly equivalent).}
and randomly select another model from the remaining models in $\M$.
\textbf{(c)} Run Reconcile on the two selected models.
\textbf{(d)} Repeat steps (a) to (c) until all models in $\M$ have been considered for reconciliation. 

For this experiment, we generalize the mode aggregation method to mean aggregation in the case of a regression task (Communities and Crime dataset).
We define our meta-rule based on the cross-validated accuracy score. This score is derived from a 5-fold cross-validation method, where the model's performance in each of the 5 iterations is evaluated in terms of accuracy, which measures the proportion of correct predictions made by the model compared to the actual outcomes. The cross-validated accuracy score for each model is then calculated as the mean of these five accuracy measurements. Models that achieve a mean cross-validated accuracy score between 0.65 and 0.70 satisfy our meta-rule. We create the set $\M$ by training different classes of classifiers with different hyper-parameters and choosing the ones that satisfy our meta-rule. A detailed list of these classifiers is available in Appendix \ref{appendix:models}.
After building the set $\M$, we run a prediction task on the test set using the four aggregation methods and calculate the Brier score of the final resulting model.\footnote{Reconcile was applied to the validation data.} We repeated each experiment with each data set 100 times\footnote{A boxplot comparing methods for each dataset individually is available in Appendix \ref{appendix:results}, Figure \ref{fig:aggregate_boxplot}} and performed a pairwise t-test to assess whether Sequential Reconcile performs statistically better compared to the other three methods in all data sets. The results are displayed in Table \ref{tab:MM result}. These results show that Sequential Reconcile performs significantly better than the other proposed methods if the goal is to have the highest accuracy possible. As is the case in our results, the Mean Aggregation technique can in some cases perform better than sequential Reconcile in terms of Brier score.\footnote{One hypothesis is that applying the patch to continuous values could in some cases result in passing the correct prediction (note that this cannot happen in the case of a binary classification in which as long as you are going in the right direction, one can only get closer to the correct answer).}

\begin{table}[h]
\centering
\begin{tabular}{|p{2cm}|p{1.5cm}|p{2cm}|p{1.5cm}|p{1.5cm}|p{1.5cm}|}
\hline
Method \# & T-Stat. &  P-Value & Mean Diff & CI Lower & CI \,\,\,\,\, Upper \\ \hline
1 \,\,\,\,\,\,(Mode) & $-11.478$  & $4.09 \times 10^{-29}$ & $-0.032$  & $-0.038$  & $-0.027$ \\ \hline
1 \,\,\,\,\,\,(Mean) & $21.221$   & $1.35 \times 10^{-79}$ & $0.130$   & $0.126$   & $0.135$\\ \hline
2  & $-3.667$   & $2.55 \times 10^{-4}$ & $-0.012$  & $-0.019$  & $-0.006$\\ \hline
3 & $-3.537$   & $4.18 \times 10^{-4}$ & $-0.012$  & $-0.019$  & $-0.005$ \\ \hline
\end{tabular}
\caption{Pairwise t-test results comparing different methods to Sequential Reconcile across all datasets, with 95\% confidence intervals.}
\label{tab:MM result}
\end{table}
Each of these techniques provides a justifiable way to resolve multiplicity in different contexts. 
It has been shown that the mode aggregation technique reduces instability in predictions \citep{black2021selective}. 
In \citep{black2022model} the authors argue that there are plenty of applications where randomized predictions such as the ones described by aggregation methods 2 and 3 are undesirable; but in applications where decisions are low-stakes and repeated, this randomized sampling might give a person outcomes that better reflect the uncertainty contained in the model distribution. 
However, what is important to note is that the theoretical guarantees of the Reconcile algorithm provide far more benefits than those of non-arbitrariness.

In the previous solutions for mode aggregation, the models still do not agree on the final prediction for specific data points.
This can be seen as a \emph{compromise} procedure, where despite underlying disagreements, the decision is settled by, for example, a majority vote, requiring some models to yield to the prevailing outcome.
In the case of (sequential) Reconcile, however, we are aggregating through a ``consensus-reaching'' process, in which when each pair of models is contested, they both agree that at least one of the two models is falsifiable, and they both agree that the final models are at least as good as the initial ones (in terms of Brier score) while presenting no significant disagreement between them. 
To formalize this idea of a consensus-reaching process, we formally study the robustness and fairness properties of Reconcile as a model class aggregation algorithm next.

\subsubsection{Robustness of Sequential Reconcile} \label{sec:experiment:stability}
Another way to conceptualize the benefits of sequential Reconcile compared to the other aggregation methods is its stability. Unlike previous solutions for model aggregation, where even having one slightly bad model in the set $\M$ can have a significant negative impact on the aggregated model, Reconcile is robust, meaning the final predictions remain consistent and stable as long as there is at least one model in the set that is performing well enough. 

To demonstrate this empirically, we construct the set in the same way as described in Section~\ref{sec:experiments Reconcile and MM}. In this experiment, we progressively replace the models in the set with random predictors, starting with one replacement and continuing until all models are replaced.
 We focus on the mode aggregation method (or mean aggregation for regression tasks) and sequential Reconcile, as the other two methods (randomized prediction and random model selection) inherently exhibit arbitrary behavior. For each dataset, we generate a set of models and incrementally substitute random models for actual ones. The number of models in each set varies between datasets, as determined by the meta-rule specified in Section \ref{sec:experiments Reconcile and MM}. Figure \ref{fig:stability_heatmap} displays the results for a single run, showing that sequential Reconcile consistently leads to better predictions until all models are replaced with random predictors. In contrast, mode aggregation quickly reaches its upper bound for random predictors, requiring as few as 2 or 3 bad predictors to significantly degrade performance. We conducted the experiment 20 times and consistently observed similar patterns (results available in Appendix \ref{appendix:results}, Figure \ref{fig:stability_mean}). As the number of bad predictors (random models) in the set increases, the difference in MSE also increases, until all models are replaced and the performance of Sequential Reconcile declines. These results indicate that while Sequential Reconcile may not always outperform other aggregation methods (e.g., mean aggregation) in terms of Brier score and comes with higher computational costs, its ability to produce more stable predictions provides a strong justification for its use.
The formal justification of this result beyond the empirical demonstration is given in the following statement.

\begin{lemma}\label{lemma:brierscores}
    For any pair of models $f_1, f_2: \X \rightarrow [0,1]$, any dataset $D$, and any $\alpha, \epsilon>0$, Reconcile$(f_1, f_2)$ runs for $T = T_1+T_2$ many rounds and outputs a pair of models $(f_1^{T_1}, f_2^{T_2})$ such that both $B(f_1^{T_1})$ and $B(f_2^{T_2})$ are upper bounded by 
    \[
        \min \Big\{ B(f_1) - T_1 \cdot \frac{\alpha \epsilon^2}{16}, B(f_2) - T_2 \cdot \frac{\alpha \epsilon^2}{16}\Big\} +4\epsilon + 3\alpha.
    \]
\end{lemma}

\begin{proof}
    We write the proof for an input dataset $D = \{(x_i, y_i), \ldots, (x_n, y_n)\}$, given that this corresponds to the set-up of our experiments (the proof carries forward directly for an arbitrary distribution $\D$).
    We relate the Brier scores of $f_1^{T_1}$ and $f_2^{T_2}$ by separating between the area of agreement (i.e., $x \in \X \setminus U_{\epsilon}$) and the area of disagreement (i.e., $x \in U_{\epsilon}$).
    In the agreement area, we know that $f^{T_1}_1$ and $f^{T_2}_2$ satisfy $|f^{T_1}_1(x)-f^{T_2}_2(x)|\leq \epsilon$ for all $x$, given that this condition is satisfied for $f_1, f_2$, and hence is also satisfied by $f_1 = f_1^{T_1}$, $f_2 = f_2^{T_2}$ over the agreement area, given that the Reconcile algorithms only modifies the predictions of the classifiers over the disagreement area.
    In the area of disagreement, the difference between $f^{T_1}_1$ and $f^{T_2}_2$ can be as large as 1 (since the range of the models is $[0,1]$).
    However, by the guarantee of the Reconcile algorithm (Theorem~\ref{thm:mainthm}), the disagreement area is at most an $\alpha$ fraction of the domain.
    \begin{align*}
        \big| B(f_1^{T_1}) - B(f_2^{T_2})\big| 
        &\leq  \dfrac{1}{n} \cdot \Big| \sum_{i : x_i \in \mathcal{X} \setminus U_{\epsilon}} (f_1^{T_1}(x_i) - y_i)^2 - (f_2^{T_2}(x_i) - y_i)^2
        \Big|  \\
        &\quad  + \dfrac{1}{n} \cdot \Big| \sum_{i : x_i \in U_{\epsilon}} (f_1^{T_1}(x_i) - y_i)^2 - (f_2^{T_2}(x_i) - y_i)^2 \Big|
        \\
        &\leq \dfrac{1}{n} \cdot \Big( \sum_{i : x_i \in \mathcal{X} \setminus U_{\epsilon}} \big| (f_1^{T_1})^2(x_i) - (f_2^{T_2})^2(x_i) \big|+ 2y_i \cdot \big| f^{T_2}_2(x_i) - f^{T_1}_1(x_i) \big|  \Big)\\
        &\quad + \dfrac{1}{n} \cdot \Big( \sum_{i : x_i \in U_{\epsilon}} \big| (f_1^{T_1})^2(x_i) - (f_2^{T_2})^2(x_i) \big|+ 2y_i \cdot \big| f^{T_2}_2(x_i) - f^{T_1}_1(x_i) \big|  \Big) \\
        &\leq \dfrac{(2\epsilon + 2\epsilon) n}{n} + \dfrac{3 \cdot |U_{\epsilon}|}{n}  \\
        &\leq \dfrac{4 \epsilon n + 3 \alpha n}{n} = 4 \epsilon + 3 \alpha.
    \end{align*}
\end{proof}

 In particular, if $B(f_1)$ is initially low but $B(f_2)$ is not (i.e., $B(f_1) \ll B(f_2)$), this lemma ensures that the reconciled model $f_2^{T_2}$ after calling the Reconcile algorithm on $(f_1, f_2)$ is accurate as well, since by Lemma~\ref{lemma:brierscores},  $B(f_2^{T_2}) \leq B(f_1) + 4\epsilon + 3\alpha$.
(Lemma~\ref{lemma:brierscores} shows a tighter bound because the Brier score of $f_1$ can also have improved after reconciling even if its initial Brier score was much better than that of $f_2$).

\paragraph{Order matters for sequential applications of Reconcile.}
We point out that the ordering of reconciliation matters both when we use sequential Reconcile as a model aggregation technique and for producing a new class $\M'$.
In the first case, where we build a final model $f$ by applying sequential Reconcile to $\M$, a natural question is: does the order in which we process the models in $\M$ change the Brier score of the final model $f$?
The answer, as we have observed empirically, is yes. 
For example, in the case of the Adult dataset, following the same set-up as the one described in Section~\ref{sec:experiments Reconcile and MM}, one ordering yielded a final MSE of 0.155, whereas another run yielded a final MSE of 0.146. 

On a theoretical level, as shown in \cite[Lemma 3.2]{Roth2022-sd}, each update in the Reconcile algorithm decreases the Brier score of the model by exactly $B(f_t, \D) - B(f_{t+1}, \D) = \mu(g_t) \cdot \Delta_t^2$,
where $g_t$ is one of $U^{>}(f_1, f_2), U^{<}(f_1, f_2)$.
The reason for why the final improvement of Brier scores is of $T_1 \cdot O(\alpha \epsilon^2)$ and $T_2 \cdot O(\alpha \epsilon^2)$, respectively (as shown in Theorem 2.1) is because $\mu(g_t) \cdot \Delta_t^2\geq O(\alpha \epsilon^2)$ for all $t$ (see Appendix~\ref{appendix-sec:reconcilealgo} for a definition of $\Delta_t$).
Thus, the optimal ordering of the models that yields the lowest final Brier score is the one that maximizes $\sum_t^T \mu(g_t) \cdot \Delta_t^2$.

Similarly, when we use Reconcile to build a class $\M'$ from $\M$ in order to maximize the disagreement between models in $\M'$ (as measured by the different metrics that we discuss in Section~\ref{sec:experiments Reconcile and MM}), the final numbers also depend on the order in which we process the models (for all metrics quantifying disagreement within $\M'$).
For example, in the case of the COMPAS dataset, we find that the ambiguity and the discrepancy can differ by more than $0.1$ between some pairs of initial orderings (and similarly for the other metrics).

It is also important to note that, unlike other methods for ensembling models, Reconcile can be run on \emph{any} two input models, with no requirements whatsoever on how ``good'' the two initial models should be.

\begin{figure}
\centering
\includegraphics[width=1\columnwidth]{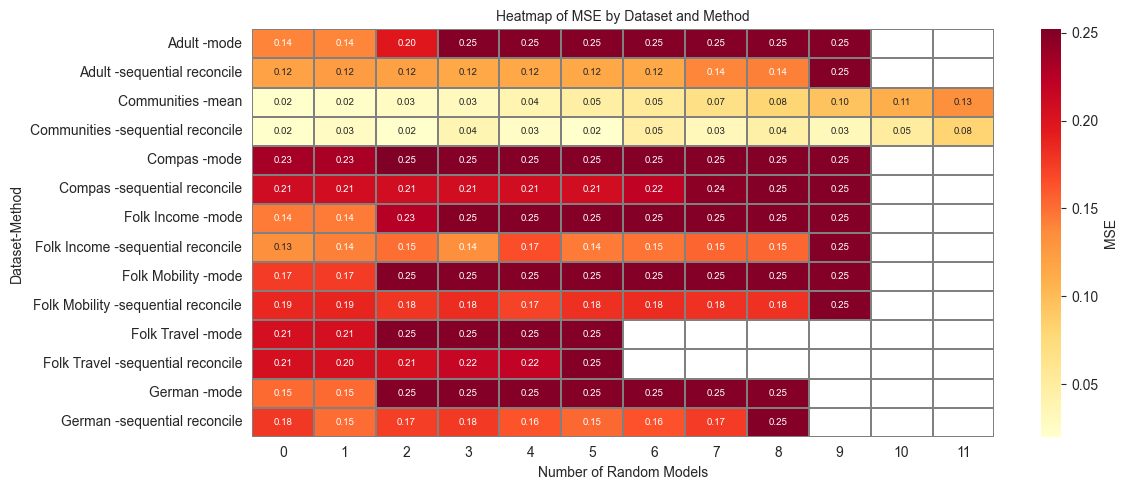} 
\caption{Heatmap showing the Mean Squared Error (MSE) across different numbers of random models for each dataset and aggregation method.}
\label{fig:stability_heatmap}
\end{figure}

\subsubsection{Fairness of Reconcile} \label{sec:fairness}
We further demonstrate the benefits of using Reconcile as a model aggregation method by studying the fairness properties of Reconcile.
Since in this work we assume that $f^*$ corresponds to the ground-truth, we view fairness as accuracy: an unfair model is one that has low Brier score over a minority group.
Here we demonstrate that if a model is fair on a (large enough, as measured by $\alpha$) minority group, then reconciling it with a model that is unfair over the group will ensure that the resulting reconciled model is more accurate on the minority group.
 
Formally, let $B(f)|_{P}$ denote the Brier score of predictor $f$ restricted to subset $P \subseteq \X$ of the domain; in fairness applications we are interested in ensuring that $B(f)|_{P}$ is low for all minority groups $P$:

\begin{lemma}\label{lemma:fairness}
    Let $P \subseteq \X$ be a subgroup of the domain such that $\mu(P) \geq \sqrt{4\epsilon+3\alpha}$ for some given $\alpha>0$.
    Let $f_1, f_2: \X \rightarrow [0,1]$ be two models such that $B(f_1)|_{P} \leq B(f_2)|_{P}$, and let $f_1^{T_1}, f_2^{T_2}$ denote the final models after running Reconcile$(f_1, f_2)$.
    Then, $B(f_2^{T_2})|_P \leq B(f_1)|_{P} + \sqrt{4\epsilon +3\alpha}$.
\end{lemma}

\begin{proof}
    By Lemma~\ref{lemma:brierscores}, we know that $| B(f_1^{T_1}) - B(f_2^{T_2})| \leq 4\epsilon + 3\alpha$.
    In the worst case, this difference in Brier score is concentrated within $P$, such that 
    \[
        \big| B(f_1^{T_1})|_P - B(f_2^{T_2})\big|_P \leq \frac{4\epsilon + 3\alpha}{\mu(P)}.
    \]
    Therefore, by imposing that the minority group be somewhat large; in particular, $\mu(P) \geq \sqrt{4\epsilon +3\alpha}$, and recalling that $B(f_1^{T_1}) \leq B(f_1) - T_1 \cdot \frac{\alpha \epsilon^2}{16}$, the result follows.
\end{proof}

We note that in practice, however, it is unlikely that the few number of disagreements left after running the Reconcile algorithm are concentrated in the subgroup $P$. 

Nevertheless, if $f_1$ does initially well in subgroup $P$ but $f_2$ is highly inaccurate within $P$, reconciling these two models will result in the improvement of $f_2$'s accuracy within (minority) group $P$. 
This relevant property is not necessarily true for other aggregation methods such as, e.g., mean aggregation.
Both lemmas on the stability and fairness properties of Reconcile naturally extend to a class of models, and they generalize to the entire distribution as in \citep{Roth2022-sd}.

To empirically evaluate the fairness properties of Reconcile, we design a scenario where one of the two models performs poorly on a specific subgroup. Using the dataset's race information, we define two subgroups: the majority group, which consists of individuals from the largest racial group in terms of population size, and the minority group, which includes individuals from the second-largest racial group (these two subgroups do not encompass the entire dataset and might be different racial groups in different datasets). After training the models,\footnote{Information on the models is available in Appendix \ref{appendix:models}.} we assign random predictions to the minority subgroup for one of the models. The results of this experiment are presented in Figure \ref{fig:fairness}. As shown, consistent with Lemma~\ref{lemma:fairness}, the final predictions for both models after applying Reconcile demonstrate that the model with low accuracy on the minority group now performs almost as well as the better-performing model. This experiment shows how Reconcile contributes to improving fairness. If one of the models performs reasonably well on the minority group—for instance, if it was specifically trained on this group—the final reconciled predictions show noticeable improvement for this subgroup.

\begin{figure}
\centering
\includegraphics[width=0.9\columnwidth]{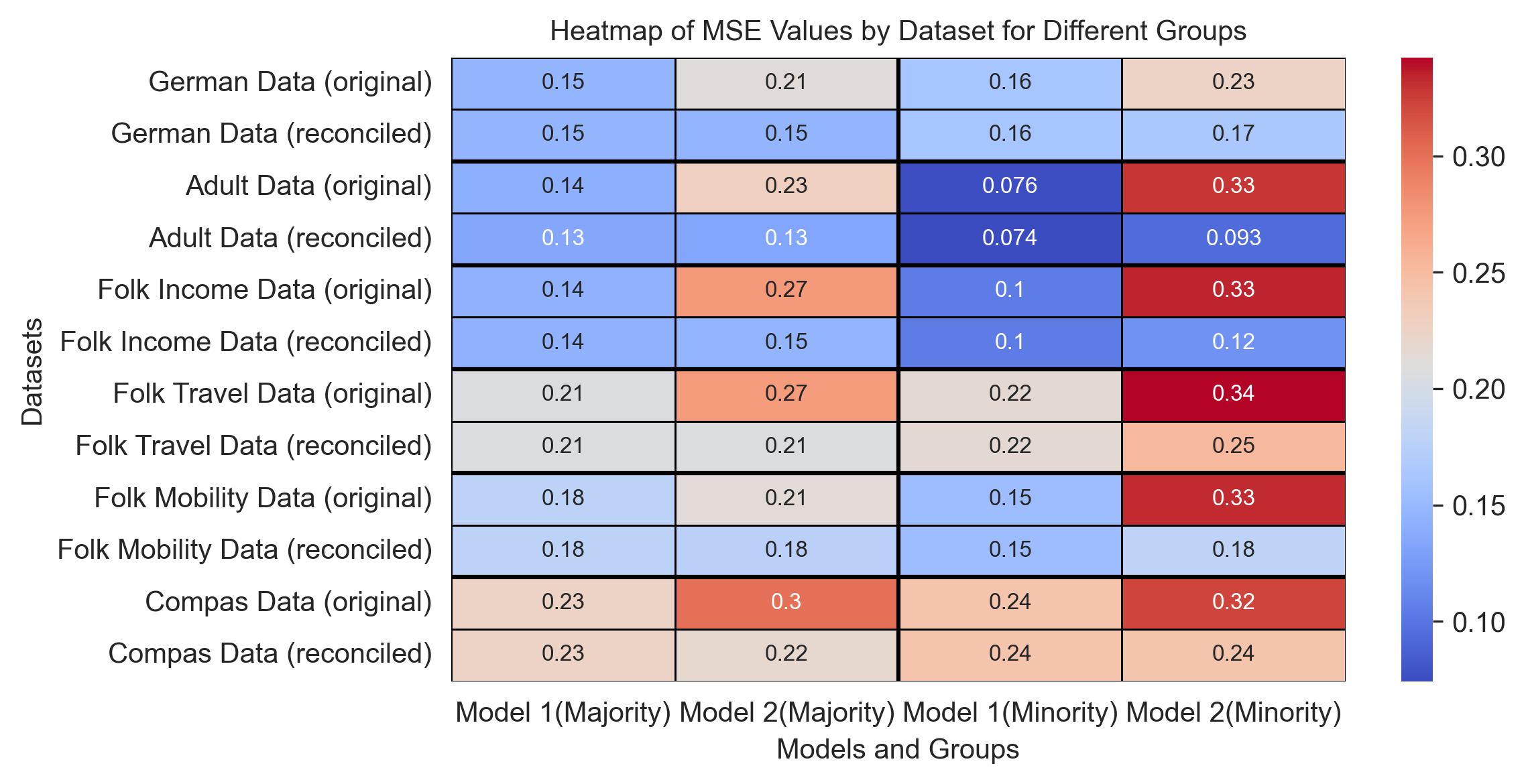} 
\caption{Heatmap showing MSE values for majority and minority race groups across two models before and after applying Reconcile, for all datasets. In this experiment, random predictions were assigned to the second model for the minority subgroup prior to Reconcile to simulate a scenario where one model underperforms on a specific subgroup.}
\label{fig:fairness}
\end{figure}

\subsection{Using Reconcile to Reduce Disagreements Within a Class $\M$} \label{sec:experiments solution}
 To be able to analyze how Reconcile performs in terms of reducing the severity of predictive multiplicity, we use the four following metrics to quantify the disagreement on predictions amongst models in $\M$.
 These metrics originate from the predictive multiplicity literature, and so analyzing how Reconcile impacts these metrics is crucial to understanding how Reconcile relates to the MM literature.

 \begin{enumerate}
     \item Variance in predictions: For each instance $x_i$ in $D$, we calculate the variance (\(\sigma^2(x_i)\)) of the predictions made by all models in the set $\M$. We then compute and report the statistics (e.g., min, max, mean, standard deviation) of these variances across the entire $D$.
     \item Ambiguity: We generalize the definition of ambiguity proposed in the predictive multiplicity literature \citep{hsu2022rashomon} to extend it to real-valued predictions, rather than limiting it to binary outcomes only.
The \emph{ambiguity} \( \alpha(x) \) for real-valued predictions is defined as the average of the maximum absolute differences in predictions across all instances: $\phi(D) = \frac{1}{n} \sum_{i=1}^{n} |\max_{f \in \mathcal{M}}\hat{y}_f(x_i) - \min_{f \in \mathcal{M}} \hat{y}_f(x_i)|$.
    \item Discrepancy: We similarly generalize the definition of discrepancy from \citep{hsu2022rashomon} to accommodate real-valued predictions. 
     The \emph{discrepancy} \( \delta(D) \) is defined as:
        $\delta(D) = \max_{(f_1, f_2) \in \mathcal{M}} \left\{ \frac{1}{n} \sum_{i=1}^{n} |\hat{y}_{f_1}(x_i) - \hat{y}_{f_2}(x_i)| \right\}$.
    \item Disagreement Set Probability Mass: For all 2-combinations of models \( f_1 \) and \( f_2 \) in $\M$, we calculate \( \mu(U_{0.2}(f_1, f_2))) \)). We then compute and report the statistics (min, max, mean, standard deviation) of these probability masses across the entire set $\M$.
 \end{enumerate}
 
 Another challenge that arises when applying Reconcile to models in the set $\M$ is determining how models from the set should be selected for reconciliation.
 Namely, recall that the Reconcile algorithm takes two models at a time as input, and thus there are various ways in which we can process the set $\M$.
 For our study, we propose the following 4 methods to create the set $\M'$ from reconciled models:
 \begin{enumerate}
 \item[a.] Run Reconcile on all possible pairs of models from $\M$. Then, \( |\M'_a| \) = \( |\M| \choose 2 \). 
 \item[b.] Run Reconcile on all possible pairs of models from $\M$. We know from Section \ref{sec:experiments Reconcile results} that in most cases the two models are identical after Reconcile so we choose only the first model after Reconcile.
 Then, \( |\M'_b| = {|\M| \choose 2}/2 \).
 \item[c.] Run Reconcile on all possible pairs of models from $\M$ then choose $|\M|$ models randomly.
 Thus, \( |\M'_c| = |\M| \).
 \item[d.] Randomly select two models from the set $\M$ and apply Reconcile to these models. Randomly choose one of the two reconciled models and add it to a new set, $\M'$. Continue the process by randomly selecting another model from the remaining models in $\M$. Run Reconcile between this selected model and the most recently added model to $\M'$. Repeat until all models in $\M$ have been chosen.
 Thus, \( |\M'_d| = |\M| \).
 \end{enumerate}

For our empirical study, we create the set $\M$ as specified in Section \ref{sec:experiments Reconcile and MM}. We define our meta-rule based on the cross-validated accuracy score and create the
set $\M$ by training different classes of classifiers with different hyper-parameters and choosing the ones that satisfy our meta-rule.  A detailed list of these models is available in Appendix \ref{appendix:models}.
We then apply Reconcile on models in $\M$ using the approaches we outlined. To understand the effectiveness of each method, we perform a Tukey's HSD (Honestly Significant Difference) \citep{tukey1949comparing} test on ambiguity and discrepancy within the sets created with each method. This test is a well-known single-step multiple comparison procedure that can be used to correctly interpret the statistical significance of the difference between means that have been selected for comparison. The results are shown in Tables \ref{tab:exp 3 Tukey ambiguity} for ambiguity.
Based on the results shown in Table \ref{tab:exp 3 Tukey ambiguity} significant differences were found between all group pairs. The p-adjusted values (p-adj) are all 0.0 \footnote{The exact p-values range from \(10^{-17}\) to \(10^{-145}\) and are rounded to 0.}, suggesting that all pairwise differences in means are statistically significant. The mean differences between the groups range from 0.038 to 0.202, with $\M$ consistently showing larger differences when compared to the other groups that went through some form of Reconcile. The direction of the mean differences suggests that $\M$ has larger values of ambiguity compared to the other groups. When comparing between groups $\M’_a$,$\M’_b$,$\M’_c$, and $\M’_d$, we can see that the lowest ambiguity values seem to be in $\M’_d$, the set that is the produced by sequential reconcile. A similar pattern is observed across the remaining three metrics we defined. We provide a detailed discussion of these metrics in Appendix \ref{appendix:results:within-class}.

\begin{table}[h]
\centering
\begin{tabular}{|c|c|c|c|c|c|}
\hline
Set 1 & Set 2 & Mean Diff & p-adj & Upper & Lower \\ \hline
$\M$   &$\M'_a$ & $0.046$ & $0.0$ & $0.057$ & $0.034$ \\ \hline
$\M$   &$\M'_b$ & $0.105$ & $0.0$ & $0.116$ & $0.093$ \\ \hline
$\M$   &$\M'_c$ & $0.143$ & $0.0$ & $0.155$ & $0.131$ \\ \hline
$\M$   &$\M'_d$ & $0.202$ & $0.0$ & $0.213$ & $0.190$ \\ \hline
$\M'_a$   &$\M'_b$ & $0.059$ & $0.0$ & $0.071$ & $0.048$ \\ \hline
$\M'_a$   &$\M'_c$ & $0.097$ & $0.0$ & $0.109$ & $0.086$ \\ \hline
$\M'_a$   &$\M'_d$ & $0.157$ & $0.0$ & $0.168$ & $0.145$ \\ \hline
$\M'_b$   &$\M'_c$ & $0.038$ & $0.0$ & $0.049$ & $0.026$ \\ \hline
$\M'_b$   &$\M'_d$ & $0.097$ & $0.0$ & $0.109$ & $0.086$ \\ \hline
$\M'_c$   &$\M'_d$ & $0.059$ & $0.0$ & $0.071$ & $0.048$ \\ \hline
\end{tabular}
\caption{Multiple Comparison of Mean values for ambiguity across all studies for different sets - Tukey HSD, FWER=0.05}
\label{tab:exp 3 Tukey ambiguity}
\end{table}

These results support the effectiveness of Reconcile in reducing model disagreement and increasing consistency across predictions of the models in set $\M$ regardless of the approach used to select models for Reconcile. Despite methods (a) and (b) increasing the number of models in the set, the reconciled models perform better in terms of predictive multiplicity. Method (c) performs slightly better than methods (a) and (b), possibly because the size of $\M'$ matches that of $\M$. Method (d), which performs Reconcile in a sequential manner reveals more consistency across predictions of models in the set while being the least computationally expensive. This highlights that the same or even higher consistency can be achieved without applying Reconcile to all possible model pairs, which scales quadratically as the set size grows. Instead, sequential Reconcile offers a more efficient alternative, scaling linearly with the set size.

\subsection{Relationship Between Reconcile and Multiaccuracy}
Lastly, we clarify the relationship between the Reconcile algorithm and the recent literature on multigroup fairness, which has attracted increasing interest in the past couple of years.
The definitions of multiaccuracy (MA) and multicalibration (MC) were first defined as a mathematical definition of fairness by \cite{hebert2018multicalibration}: given a pre-specified collection of groups $\mathcal{G}$, we want to build a predictor $f$ that satisfies a particular relationship with respect to the true probabilities $f^*$ for \emph{every} group $g \in \G$.
More specifically, one variant of the definition of a multiaccurate predictor is the following:\footnote{The definition is usually stated with $f^*$ instead of $y$, but we can relate the two expression using concentration inequalities. The definition is also usually stated with an absolute value instead of a squaring. Different papers deal with small groups $g$ in different ways: we can either condition on $g$ and then incorporate $\mu(g)$ into the error parameter, or multiply by $g$ inside of the expectation term.}
\begin{definition}\label{def:MA}
    Given a ground-truth model $f^*: \X \rightarrow [0,1]$, true labels $y \sim Bern(f^*)$, a distribution $\D$ on $\X \times \{0, 1\}$, a collection of groups $\G$ on $\X$,
    and $\beta>0$, we say that a predictor $f: \X \rightarrow [0,1]$ is $\beta$-\emph{multiaccurate} if, for all $g \in \G$,
    \[
        \Big( \E_{(x, y) \sim \D}[f(x) - y | g(x) = 1] \Big)^2 \leq \dfrac{\beta}{\mu(g)}.
    \]
\end{definition}

Multiaccuracy is a definition; what \cite{hebert2018multicalibration} first showed is that, given any collection of groups $\G$, ground-truth $f^*$, and $\beta>0$, we can always build a $\beta$-MA predictor efficiently (and similarly for the stronger notion of multicalibration).
Their algorithm is as follows: we start with the trivial predictor $f=0$, and we iteratively find some group $g \in \G$ that witnesses a violation of $\beta$-MA. 
Next, we update the predictor $f$ on all the points $x$ such that $g(x)=1$ by patching it with the witness $g$ as follows:
\[
    f(x) \leftarrow f(x) + \Delta, \,\, 
    \Delta := \E_{(x, y) \sim \D}[f(x) - y | g(x) = 1].
\]
The MA algorithm terminates when no more groups in $g \in \G$ that witness a MA violation can be found.
In order to show that the MA algorithm terminates within not too many steps, we use a potential argument with the potential function $\Phi = \E_{\D}[(f^*-f_t)^2]$ and show that each patching iteration improves $\Phi$ by at least $O(\beta^2)$. 
Thus, the algorithm terminates within $O(1/\beta^2)$ iterations.

The MA algorithm is related to the Reconcile algorithm as follows.
The Reconcile algorithm follows exactly the same paradigm as the MA algorithm: find a group that witnesses a violation of the guarantee that we are trying to satisfy, patch the current $f_t$ by updating the values $x$ inside the witnessing group by adding $\Delta$, and then show that this update has improved the Brier score of $f$ by at least some quantity that depends on the initial parameters.
(Note that the potential function that is used in the MA algorithms corresponds precisely to the Brier score of the model $f$.)

In the case of Reconcile, we do not have a pre-specified collection of groups $\G$. 
Instead, at each step of the algorithm, a group gets naturally defined: namely, the subset $U^{>}_{\epsilon}, U^{>}_{\epsilon}$ that witnesses a violation of $\alpha$-group conditional mean consistency.
As in the case of MA, the Reconcile algorithm works because we can show that patching $f_t$ with $\Delta$ yields an improvement of the Brier score of at least $O(\alpha \epsilon^2)$.
As in MA, Reconcile only patches the model inside of the witnessing group $g = U^{>}_{\epsilon}, U^{<}_{\epsilon}$.\footnote{Indeed, Reconcile never modifies the values of either $f_1$ or $f_2$ in the area where the two models already agree from the beginning; i.e., in $\X \setminus U_{\epsilon}$. Thus, the improvement of the Brier score in the models comes entirely from the initially disagreeing area. This is also why Reconcile only yields MA with respect to a collection of groups that are all defined over the initial $U_{\epsilon}$, as shown in Lemma~\ref{lemma:MAReconcile}.}
Note that the subsets $U_{\epsilon}$ evolve at each iteration $t$; we can thus write $U_{\epsilon, t}$.
A full run of the Reconcile thus defines a collection of groups $\G$ on $\X$ given by the violation-witnessing group at each iteration $t$; i.e., $\G = \{U^{>, <}_{\epsilon, t}\}_{t \in T}$.
This 1-to-1 matching between the MA algorithm and the Reconcile algorithm thus leads to the following formal relationship:

\begin{lemma}[Relationship between MA and Reconcile]\label{lemma:MAReconcile}
    Given a ground-truth model $f^*: \X \rightarrow [0,1]$, true labels $y \sim Bern(f^*)$, a distribution $\D$ on $\X \times \{0,1\}$, a pair of models $f_1, f_2: \X \rightarrow [0,1]$, and parameters $\alpha, \epsilon >0$, let $f$ correspond to the output of Reconcile$(f_1, f_2)$ after $T$ rounds with parameters $\alpha, \epsilon, \D$.
    Then, $f$ is $(\alpha\epsilon^2)$-multiaccurate with respect to the collection of groups $\G = \{U^{>, <}_{\epsilon, t}\}_{t \in T}$.
\end{lemma}

\section{Extending Reconcile to the Causal Inference Setting}
\label{sec:prelims}

In order to further demonstrate the applicability and effectiveness of the reconciliation paradigm, we now extend the original Reconcile algorithm to the setting of causal inference, which we call \emph{ReconcileCATE}. 
Similar to the usual PM problem, we can obtain different causal average treatment effect (CATE) estimators $\hat{\tau}$ for the same dataset $D = \{(X_i, Y_i, T_i)\}$ that disagree on certain predictions $\hat{\tau}(x)$.
Thus, Reconcile algorithm offers a robust solution to both predictive multiplicity in machine learning and heterogeneous treatment effects in causal inference. 
We provide the first experiments concerning the application of Reconcile to the setting of Causal Inference.
We can solve the ReconcileCATE problem by using the Reconcile algorithm as a subroutine.
Formally:
\begin{lemma}\label{lemma:reconcile-cate-to-reconcile}
    Given two CATE predictors $\hat{\tau}_1, \hat{\tau}_2$ trained on the same dataset $D = \{(X_i, Y_i, T_i)\}$, and parameters $\alpha, \epsilon >0$ we can solve the ReconcileCATE problem efficiently with one call to the Reconcile algorithm.
\end{lemma}

\begin{proof}
    Given the causal dataset $D = \{(X_i, Y_i, T_i)\}$, we transform it into a new dataset $\mathcal{D'} = \{(x_i, \hat{y}_i)\}_{i=1}^N$, where $\hat{y}_i$ for each $x_i \in \mathcal{X}$ is defined as $\hat{y}_i = \mathbb{E}[y | X = x_i, T = 1] - \mathbb{E}[y | X = x_i, T = 0]$.
    We call the Reconcile algorithm (Algorithm~\ref{alg:reconciler}) setting $f_1 := \hat{\tau}_1, f_2 = \hat{\tau}$ and dataset $\mathcal{D'}$. 
    By the guarantees of Theorem~\ref{thm:mainthm}, from the call Reconcile$(f_1, f_2)$ we obtain two predictors $f_1^{T_1}, f_2^{T_2}$ such that $\mu(U_{\epsilon}(f_1^{T_1}, f_2^{T_2})) < \alpha$. 
    Thus, the same guarantee holds for $\hat{\tau}_1, \hat{\tau}_2$.
 \end{proof}

Hence we obtain the same guarantees of quick convergence, Brier score improvements, and final disagreement between CATE estimators.
As established in the proof of Lemma~\ref{lemma:MAReconcile}, ReconcileCATE ensures that CATE predictors are multiaccurate (MA) with respect to the groups $\G = \{U^{>, <}_{\epsilon, t}\}_{t \in T}$ defined by the disagreement region $U_{\epsilon}(\hat{\tau}_1, \hat{\tau}_2)$. This region includes subsets of the covariate space where the initial models $\hat{\tau}_1$ and $\hat{\tau}_2$ differ significantly, i.e., where $|\hat{\tau}_1(x) - \hat{\tau}_2(x)| \geq \epsilon$. By enforcing multiaccuracy within $U_{\epsilon}$, ReconcileCATE adjusts CATE estimates to better align with the observed data in regions of high model disagreement.

Doubly-Robust (DR) estimation for CATE refers to the property that an estimator is consistent even when one of two models is misspecified \citep{kennedy2023towards}.
Specifically, the estimator is consistent if either the propensity score model $\hat{e}(x)$ or the outcome regression models $\hat{\mu}_1(x)$ and $\hat{\mu}_0(x)$ are correctly specified. The outcome regression models represent the expected outcomes for the treated and control groups, given covariates $x$: $\mu_1(x) = \mathbb{E}[Y \mid X = x, T = 1]$, $\mu_0(x) = \mathbb{E}[Y \mid X = x, T = 0]$.
These models describe the relationship between covariates and outcomes for treated and untreated populations, respectively.

Recent work by \cite{kern2024multi} shows how we can post-process CATE estimators with multiaccuracy to ensure doubly-robustness.
Thus, since by Lemma~\ref{lemma:MAReconcile} ReconcileCATE achieves MA with respect to a collection of groups within $U_{\epsilon}$, in our setting, we could hope to obtain doubly-robustness as a by-product of running ReconcileCATE (due to the MA guarantee of ReconcileCATE).
However, \citep{kern2024multi} require a large and rich class of auditors $\mathcal{F}$, whereas the collection $\G$ in the case of ReconcileCATE (in Lemma~\ref{lemma:MAReconcile}) falls short.
Still, future research may investigate ways of enriching the collection $\G$ so that we can obtain strong doubly-robustness guarantees from ReconcileCATE, or to further exploit the underlying multiaccuracy guarantees.

For example, by Lemma~\ref{lemma:MAReconcile}, which shows how Reconcile yields a multiaccuracy guarantee, we can apply the Reconcile algorithm to CATE predictors that have been post-processed with a multiaccuracy procedure as per Algorithm 1 in \citep{kern2024multi}, which ensures robustness under unknown covariate shifts, and the same guarantee will be true after the reconciliation procedure.
That is, given two $T$-learner estimates $\hat{\tau}_1, \hat{\tau}_2$, suppose we apply Algorithm 1 from \citep{kern2024multi} to $\hat{\tau}_1, \hat{\tau}_2$, obtaining $\tilde{\tau}_1, \tilde{\tau}_2$ that satisfy Proposition 2 in \citep{kern2024multi}.
Then, we can call the ReconcileCATE$(\tilde{\tau}_1, \tilde{\tau}_2)$ algorithm, obtaining predictors $(\tilde{\tau}_1^{T_1}, \tilde{\tau}_2^{T_2})$.
Given that by Lemma~\ref{lemma:MAReconcile}
ReconcileCATE only expands the multiaccuracy guarantee to a larger collection of groups, both predictors $\tilde{\tau}^{T_1}_1, \tilde{\tau}^{T_2}_2$ continue to satisfy the robustness guarantee under unknown covariate shifts shown in \citep[Proposition 2]{kern2024multi}.

\section{Evaluation of ReconcileCATE} \label{sec:CATE experiments}
\subsection{Building Pairs of Estimators to Reconcile}
For our experiments in this section, we use the Twins dataset \citep{almond2005costs, guo2020survey} and the National Study dataset \citep{nosek2015promoting}. Further details on the data sets can be found in Appendix \ref{appendix:datasets}.  For implementing causal estimators, we use the CausalML Python package \citep{chen2020causalml}.
Our approach here is similar to the one described in Section \ref{sec:building models} with some small changes.
First, we need to define the subgroups that we want to condition on for estimating the average effect. To do so, we train a causal tree first and use the leaf nodes as our subgroups. With this approach, we end up with $468$ and $173$ subgroups on Twins and National Study datasets.
We then move on to training causal estimators that have comparable accuracy but significant disagreement between their estimates for our groups (the same definition as in Section~\ref{sec:building models} but with $\epsilon = 0.04$ and $\alpha = 0.01$). For our causal experiments, we only use the first and second methods from Section \ref{sec:building models} (different estimators and different subsets of training data).
We use R, S, T, and X learners as well as a causal uplift tree for our estimators. Details on each of these estimators can be found in Appendix \ref{appendix:models}.
After training estimators with the desired criteria, we use the $468$ and $173$ CATE estimates as the data points for Reconcile. For each method and each estimator, we repeat the experiment 20 times.
\begin{figure}
    \centering
    \includegraphics[width=15cm]{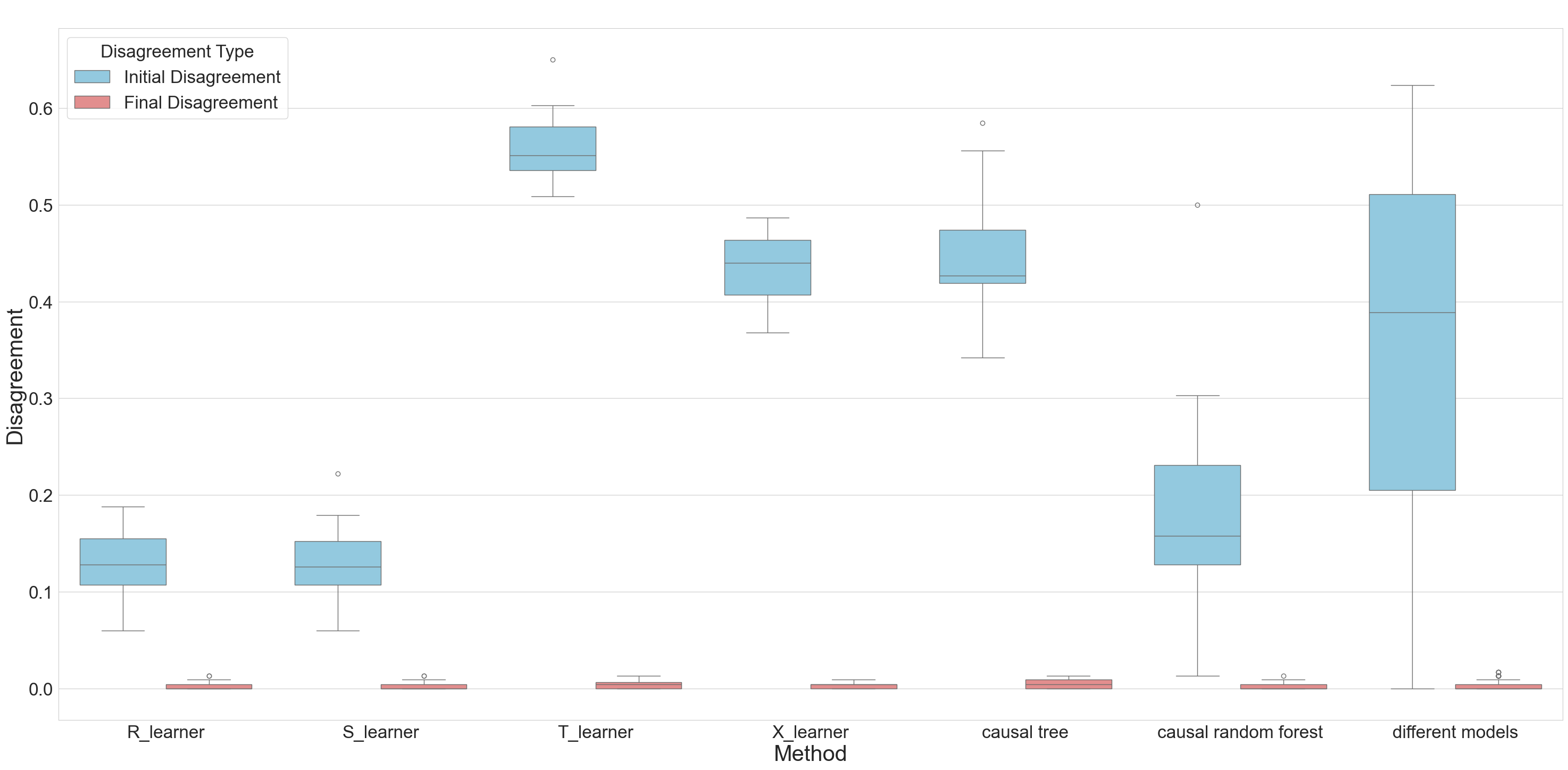}
    \caption{Disagreements between CATE estimates before and after running Reconcile on Twins dataset.}
    \label{fig:CATE_disagreement_twin}
\end{figure}

\subsection{Results}
Across all experiments, the algorithm converged within 1 to 18 rounds. As shown in Figure \ref{fig:CATE_disagreement_twin} Reconcile effectively reduces disagreement to zero in nearly all cases (the same is true for the National Study dataset and is illustrated in Figure \ref{fig:CATE_disagreement_national} in Appendix \ref{appendix:results}). In these plots, the \textit{different models} refer to the method in which we use two different estimators (among the 5 different estimators mentioned) to find estimators that have disagreement (second methods from Section \ref{sec:building models}). Additionally, when examining the Brier scores of the estimators before and after applying Reconcile, we observe consistent improvements (results can be found in Appendix \ref{appendix:results}, Figures \ref{fig:CATE_brier_twin} and \ref{fig:CATE_brier_national}).

\section*{Acknowledgments}
We thank Dylan Small and Ilya Shpitser for helpful discussions on causal connections. We are also grateful to Aaron Roth, Michael Kearns, Zhiwei Steven Wu, Christopher Jung, Jamelle Watson-Daniels, and Flavio P. Calmon for valuable feedback, comments, and discussions. 
We thank Christoph Kern and Angela Zhou for answering our questions.
Finally, we warmly thank Audrey Chang for her contributions to a related early-stage project and her efforts under guidance, which have informed some of the present work.

\bibliographystyle{ACM-Reference-Format}
\bibliography{refs}


\begin{thebibliography}{47}


\ifx \showCODEN    \undefined \def \showCODEN     #1{\unskip}     \fi
\ifx \showDOI      \undefined \def \showDOI       #1{#1}\fi
\ifx \showISBNx    \undefined \def \showISBNx     #1{\unskip}     \fi
\ifx \showISBNxiii \undefined \def \showISBNxiii  #1{\unskip}     \fi
\ifx \showISSN     \undefined \def \showISSN      #1{\unskip}     \fi
\ifx \showLCCN     \undefined \def \showLCCN      #1{\unskip}     \fi
\ifx \shownote     \undefined \def \shownote      #1{#1}          \fi
\ifx \showarticletitle \undefined \def \showarticletitle #1{#1}   \fi
\ifx \showURL      \undefined \def \showURL       {\relax}        \fi
\providecommand\bibfield[2]{#2}
\providecommand\bibinfo[2]{#2}
\providecommand\natexlab[1]{#1}
\providecommand\showeprint[2][]{arXiv:#2}

\bibitem[Almond et~al\mbox{.}(2005)]%
        {almond2005costs}
\bibfield{author}{\bibinfo{person}{Douglas Almond}, \bibinfo{person}{Kenneth~Y Chay}, {and} \bibinfo{person}{David~S Lee}.} \bibinfo{year}{2005}\natexlab{}.
\newblock \showarticletitle{The costs of low birth weight}.
\newblock \bibinfo{journal}{\emph{The Quarterly Journal of Economics}} \bibinfo{volume}{120}, \bibinfo{number}{3} (\bibinfo{year}{2005}), \bibinfo{pages}{1031--1083}.
\newblock


\bibitem[Athey and Wager(2019)]%
        {athey2019estimating}
\bibfield{author}{\bibinfo{person}{Susan Athey} {and} \bibinfo{person}{Stefan Wager}.} \bibinfo{year}{2019}\natexlab{}.
\newblock \showarticletitle{Estimating treatment effects with causal forests: An application}.
\newblock \bibinfo{journal}{\emph{Observational studies}} \bibinfo{volume}{5}, \bibinfo{number}{2} (\bibinfo{year}{2019}), \bibinfo{pages}{37--51}.
\newblock


\bibitem[Becker and Kohavi(1996)]%
        {misc_adult_2}
\bibfield{author}{\bibinfo{person}{Barry Becker} {and} \bibinfo{person}{Ronny Kohavi}.} \bibinfo{year}{1996}\natexlab{}.
\newblock \bibinfo{title}{{Adult}}.
\newblock \bibinfo{howpublished}{UCI Machine Learning Repository}.
\newblock
\newblock
\shownote{{DOI}: https://doi.org/10.24432/C5XW20}.


\bibitem[Black et~al\mbox{.}(2024)]%
        {black2024less}
\bibfield{author}{\bibinfo{person}{Emily Black}, \bibinfo{person}{John~Logan Koepke}, \bibinfo{person}{Pauline~T Kim}, \bibinfo{person}{Solon Barocas}, {and} \bibinfo{person}{Mingwei Hsu}.} \bibinfo{year}{2024}\natexlab{}.
\newblock \showarticletitle{Less discriminatory algorithms}.
\newblock \bibinfo{journal}{\emph{Geo. LJ}}  \bibinfo{volume}{113} (\bibinfo{year}{2024}), \bibinfo{pages}{53}.
\newblock


\bibitem[Black et~al\mbox{.}(2021)]%
        {black2021selective}
\bibfield{author}{\bibinfo{person}{Emily Black}, \bibinfo{person}{Klas Leino}, {and} \bibinfo{person}{Matt Fredrikson}.} \bibinfo{year}{2021}\natexlab{}.
\newblock \showarticletitle{Selective ensembles for consistent predictions}.
\newblock \bibinfo{journal}{\emph{arXiv preprint arXiv:2111.08230}} (\bibinfo{year}{2021}).
\newblock


\bibitem[Black et~al\mbox{.}(2022)]%
        {black2022model}
\bibfield{author}{\bibinfo{person}{Emily Black}, \bibinfo{person}{Manish Raghavan}, {and} \bibinfo{person}{Solon Barocas}.} \bibinfo{year}{2022}\natexlab{}.
\newblock \showarticletitle{Model multiplicity: Opportunities, concerns, and solutions}. In \bibinfo{booktitle}{\emph{Proceedings of the 2022 ACM Conference on Fairness, Accountability, and Transparency}}. \bibinfo{pages}{850--863}.
\newblock


\bibitem[Bolinger(2020)]%
        {bolinger2020rational}
\bibfield{author}{\bibinfo{person}{Ren{\'e}e~Jorgensen Bolinger}.} \bibinfo{year}{2020}\natexlab{}.
\newblock \showarticletitle{The rational impermissibility of accepting (some) racial generalizations}.
\newblock \bibinfo{journal}{\emph{Synthese}} \bibinfo{volume}{197}, \bibinfo{number}{6} (\bibinfo{year}{2020}), \bibinfo{pages}{2415--2431}.
\newblock


\bibitem[Breiman(2003)]%
        {breiman2003statistical}
\bibfield{author}{\bibinfo{person}{Leo Breiman}.} \bibinfo{year}{2003}\natexlab{}.
\newblock \showarticletitle{Statistical modeling: The two cultures}.
\newblock \bibinfo{journal}{\emph{Quality control and applied statistics}} \bibinfo{volume}{48}, \bibinfo{number}{1} (\bibinfo{year}{2003}), \bibinfo{pages}{81--82}.
\newblock


\bibitem[Chen et~al\mbox{.}(2020)]%
        {chen2020causalml}
\bibfield{author}{\bibinfo{person}{Huigang Chen}, \bibinfo{person}{Totte Harinen}, \bibinfo{person}{Jeong-Yoon Lee}, \bibinfo{person}{Mike Yung}, {and} \bibinfo{person}{Zhenyu Zhao}.} \bibinfo{year}{2020}\natexlab{}.
\newblock \showarticletitle{Causalml: Python package for causal machine learning}.
\newblock \bibinfo{journal}{\emph{arXiv preprint arXiv:2002.11631}} (\bibinfo{year}{2020}).
\newblock


\bibitem[Dahabreh et~al\mbox{.}(2017)]%
        {dahabreh2017heterogeneity}
\bibfield{author}{\bibinfo{person}{Issa~J Dahabreh}, \bibinfo{person}{Thomas~A Trikalinos}, \bibinfo{person}{David~M Kent}, {and} \bibinfo{person}{Christopher~H Schmid}.} \bibinfo{year}{2017}\natexlab{}.
\newblock \showarticletitle{Heterogeneity of treatment effects}.
\newblock In \bibinfo{booktitle}{\emph{Methods in comparative effectiveness research}}. \bibinfo{publisher}{Chapman and Hall/CRC}, \bibinfo{pages}{247--292}.
\newblock


\bibitem[Dawid(2017)]%
        {dawid2017individual}
\bibfield{author}{\bibinfo{person}{Philip Dawid}.} \bibinfo{year}{2017}\natexlab{}.
\newblock \showarticletitle{On individual risk}.
\newblock \bibinfo{journal}{\emph{Synthese}} \bibinfo{volume}{194}, \bibinfo{number}{9} (\bibinfo{year}{2017}), \bibinfo{pages}{3445--3474}.
\newblock


\bibitem[Ding et~al\mbox{.}(2021)]%
        {ding2021retiring}
\bibfield{author}{\bibinfo{person}{Frances Ding}, \bibinfo{person}{Moritz Hardt}, \bibinfo{person}{John Miller}, {and} \bibinfo{person}{Ludwig Schmidt}.} \bibinfo{year}{2021}\natexlab{}.
\newblock \showarticletitle{Retiring adult: New datasets for fair machine learning}.
\newblock \bibinfo{journal}{\emph{Advances in neural information processing systems}}  \bibinfo{volume}{34} (\bibinfo{year}{2021}), \bibinfo{pages}{6478--6490}.
\newblock


\bibitem[Ding(2024)]%
        {ding2024first}
\bibfield{author}{\bibinfo{person}{Peng Ding}.} \bibinfo{year}{2024}\natexlab{}.
\newblock \bibinfo{booktitle}{\emph{A first course in causal inference}}.
\newblock \bibinfo{publisher}{CRC Press}.
\newblock


\bibitem[Dong and Rudin(2019)]%
        {dong2019variable}
\bibfield{author}{\bibinfo{person}{Jiayun Dong} {and} \bibinfo{person}{Cynthia Rudin}.} \bibinfo{year}{2019}\natexlab{}.
\newblock \showarticletitle{Variable importance clouds: A way to explore variable importance for the set of good models}.
\newblock \bibinfo{journal}{\emph{arXiv preprint arXiv:1901.03209}} (\bibinfo{year}{2019}).
\newblock


\bibitem[Donnelly et~al\mbox{.}(2023)]%
        {donnelly2023rashomon}
\bibfield{author}{\bibinfo{person}{Jon Donnelly}, \bibinfo{person}{Srikar Katta}, \bibinfo{person}{Cynthia Rudin}, {and} \bibinfo{person}{Edward Browne}.} \bibinfo{year}{2023}\natexlab{}.
\newblock \showarticletitle{The rashomon importance distribution: Getting rid of unstable, single model-based variable importance}.
\newblock \bibinfo{journal}{\emph{Advances in Neural Information Processing Systems}}  \bibinfo{volume}{36} (\bibinfo{year}{2023}), \bibinfo{pages}{6267--6279}.
\newblock


\bibitem[Du et~al\mbox{.}(2024)]%
        {du2024reconciling}
\bibfield{author}{\bibinfo{person}{Ally~Yalei Du}, \bibinfo{person}{Dung~Daniel Ngo}, {and} \bibinfo{person}{Zhiwei~Steven Wu}.} \bibinfo{year}{2024}\natexlab{}.
\newblock \showarticletitle{Reconciling Model Multiplicity for Downstream Decision Making}.
\newblock \bibinfo{journal}{\emph{arXiv preprint arXiv:2405.19667}} (\bibinfo{year}{2024}).
\newblock


\bibitem[Fisher et~al\mbox{.}(2019)]%
        {fisher2019all}
\bibfield{author}{\bibinfo{person}{Aaron Fisher}, \bibinfo{person}{Cynthia Rudin}, {and} \bibinfo{person}{Francesca Dominici}.} \bibinfo{year}{2019}\natexlab{}.
\newblock \showarticletitle{All models are wrong, but many are useful: Learning a variable's importance by studying an entire class of prediction models simultaneously}.
\newblock \bibinfo{journal}{\emph{Journal of Machine Learning Research}} \bibinfo{volume}{20}, \bibinfo{number}{177} (\bibinfo{year}{2019}), \bibinfo{pages}{1--81}.
\newblock


\bibitem[Gardiner(2018)]%
        {gardiner2018evidentialism}
\bibfield{author}{\bibinfo{person}{Georgi Gardiner}.} \bibinfo{year}{2018}\natexlab{}.
\newblock \showarticletitle{Evidentialism and moral encroachment}.
\newblock \bibinfo{journal}{\emph{Believing in accordance with the evidence: New essays on evidentialism}} (\bibinfo{year}{2018}), \bibinfo{pages}{169--195}.
\newblock


\bibitem[Gopalan et~al\mbox{.}(2023)]%
        {gopalan2023characterizing}
\bibfield{author}{\bibinfo{person}{Parikshit Gopalan}, \bibinfo{person}{Michael~P Kim}, {and} \bibinfo{person}{Omer Reingold}.} \bibinfo{year}{2023}\natexlab{}.
\newblock \showarticletitle{Characterizing notions of omniprediction via multicalibration}.
\newblock \bibinfo{journal}{\emph{arXiv e-prints}} (\bibinfo{year}{2023}), \bibinfo{pages}{arXiv--2302}.
\newblock


\bibitem[Guo et~al\mbox{.}(2020)]%
        {guo2020survey}
\bibfield{author}{\bibinfo{person}{Ruocheng Guo}, \bibinfo{person}{Lu Cheng}, \bibinfo{person}{Jundong Li}, \bibinfo{person}{P~Richard Hahn}, {and} \bibinfo{person}{Huan Liu}.} \bibinfo{year}{2020}\natexlab{}.
\newblock \showarticletitle{A survey of learning causality with data: Problems and methods}.
\newblock \bibinfo{journal}{\emph{ACM Computing Surveys (CSUR)}} \bibinfo{volume}{53}, \bibinfo{number}{4} (\bibinfo{year}{2020}), \bibinfo{pages}{1--37}.
\newblock


\bibitem[H{\'e}bert-Johnson et~al\mbox{.}(2018)]%
        {hebert2018multicalibration}
\bibfield{author}{\bibinfo{person}{Ursula H{\'e}bert-Johnson}, \bibinfo{person}{Michael Kim}, \bibinfo{person}{Omer Reingold}, {and} \bibinfo{person}{Guy Rothblum}.} \bibinfo{year}{2018}\natexlab{}.
\newblock \showarticletitle{Multicalibration: Calibration for the (computationally-identifiable) masses}. In \bibinfo{booktitle}{\emph{International Conference on Machine Learning}}. PMLR, \bibinfo{pages}{1939--1948}.
\newblock


\bibitem[Hofmann(1994a)]%
        {statlog_(german_credit_data)_144}
\bibfield{author}{\bibinfo{person}{Hans Hofmann}.} \bibinfo{year}{1994}\natexlab{a}.
\newblock \bibinfo{title}{{Statlog (German Credit Data)}}.
\newblock \bibinfo{howpublished}{UCI Machine Learning Repository}.
\newblock
\newblock
\shownote{{DOI}: https://doi.org/10.24432/C5NC77}.


\bibitem[Hofmann(1994b)]%
        {hofmann1994statlog}
\bibfield{author}{\bibinfo{person}{Hans Hofmann}.} \bibinfo{year}{1994}\natexlab{b}.
\newblock \showarticletitle{Statlog (german credit data) data set}.
\newblock \bibinfo{journal}{\emph{UCI Repository of Machine Learning Databases}}  \bibinfo{volume}{53} (\bibinfo{year}{1994}).
\newblock


\bibitem[Hsu and Calmon(2022)]%
        {hsu2022rashomon}
\bibfield{author}{\bibinfo{person}{Hsiang Hsu} {and} \bibinfo{person}{Flavio Calmon}.} \bibinfo{year}{2022}\natexlab{}.
\newblock \showarticletitle{Rashomon capacity: A metric for predictive multiplicity in classification}.
\newblock \bibinfo{journal}{\emph{Advances in Neural Information Processing Systems}}  \bibinfo{volume}{35} (\bibinfo{year}{2022}), \bibinfo{pages}{28988--29000}.
\newblock


\bibitem[Kennedy(2023)]%
        {kennedy2023towards}
\bibfield{author}{\bibinfo{person}{Edward~H Kennedy}.} \bibinfo{year}{2023}\natexlab{}.
\newblock \showarticletitle{Towards optimal doubly robust estimation of heterogeneous causal effects}.
\newblock \bibinfo{journal}{\emph{Electronic Journal of Statistics}} \bibinfo{volume}{17}, \bibinfo{number}{2} (\bibinfo{year}{2023}), \bibinfo{pages}{3008--3049}.
\newblock


\bibitem[Kent and Shah(2012)]%
        {Kent2012-jq}
\bibfield{author}{\bibinfo{person}{David~M Kent} {and} \bibinfo{person}{Nilay~D Shah}.} \bibinfo{year}{2012}\natexlab{}.
\newblock \showarticletitle{Risk models and patient-centered evidence: should physicians expect one right answer?}
\newblock \bibinfo{journal}{\emph{JAMA}} \bibinfo{volume}{307}, \bibinfo{number}{15} (\bibinfo{date}{April} \bibinfo{year}{2012}), \bibinfo{pages}{1585--1586}.
\newblock
\showISSN{0098-7484, 1538-3598}
\urldef\tempurl%
\url{https://doi.org/10.1001/jama.2012.469}
\showDOI{\tempurl}


\bibitem[Kern et~al\mbox{.}(2024)]%
        {kern2024multi}
\bibfield{author}{\bibinfo{person}{Christoph Kern}, \bibinfo{person}{Michael Kim}, {and} \bibinfo{person}{Angela Zhou}.} \bibinfo{year}{2024}\natexlab{}.
\newblock \showarticletitle{Multi-CATE: Multi-Accurate Conditional Average Treatment Effect Estimation Robust to Unknown Covariate Shifts}.
\newblock \bibinfo{journal}{\emph{arXiv preprint arXiv:2405.18206}} (\bibinfo{year}{2024}).
\newblock


\bibitem[Kleinberg and Raghavan(2021)]%
        {kleinberg2021algorithmic}
\bibfield{author}{\bibinfo{person}{Jon Kleinberg} {and} \bibinfo{person}{Manish Raghavan}.} \bibinfo{year}{2021}\natexlab{}.
\newblock \showarticletitle{Algorithmic monoculture and social welfare}.
\newblock \bibinfo{journal}{\emph{Proceedings of the National Academy of Sciences}} \bibinfo{volume}{118}, \bibinfo{number}{22} (\bibinfo{year}{2021}), \bibinfo{pages}{e2018340118}.
\newblock


\bibitem[K{\"u}nzel et~al\mbox{.}(2019)]%
        {kunzel2019metalearners}
\bibfield{author}{\bibinfo{person}{S{\"o}ren~R K{\"u}nzel}, \bibinfo{person}{Jasjeet~S Sekhon}, \bibinfo{person}{Peter~J Bickel}, {and} \bibinfo{person}{Bin Yu}.} \bibinfo{year}{2019}\natexlab{}.
\newblock \showarticletitle{Metalearners for estimating heterogeneous treatment effects using machine learning}.
\newblock \bibinfo{journal}{\emph{Proceedings of the national academy of sciences}} \bibinfo{volume}{116}, \bibinfo{number}{10} (\bibinfo{year}{2019}), \bibinfo{pages}{4156--4165}.
\newblock


\bibitem[Larson et~al\mbox{.}(2016)]%
        {Larson2016a}
\bibfield{author}{\bibinfo{person}{Julia Larson}, \bibinfo{person}{Julia Angwin}, \bibinfo{person}{Lauren Kirchner}, {and} \bibinfo{person}{Surya Mattu}.} \bibinfo{year}{2016}\natexlab{}.
\newblock \bibinfo{title}{How We Analyzed the COMPAS Recidivism Algorithm}.
\newblock \bibinfo{howpublished}{ProPublica}.
\newblock
\urldef\tempurl%
\url{https://www.propublica.org/article/how-we-analyzed-the-compas-recidivism-algorithm}
\showURL{%
\tempurl}


\bibitem[Laufer et~al\mbox{.}(2024)]%
        {laufer2024fundamental}
\bibfield{author}{\bibinfo{person}{Benjamin Laufer}, \bibinfo{person}{Manisch Raghavan}, {and} \bibinfo{person}{Solon Barocas}.} \bibinfo{year}{2024}\natexlab{}.
\newblock \showarticletitle{Fundamental Limits in the Search for Less Discriminatory Algorithms--and How to Avoid Them}.
\newblock \bibinfo{journal}{\emph{arXiv preprint arXiv:2412.18138}} (\bibinfo{year}{2024}).
\newblock


\bibitem[Lincoln and Skrzypek(1989)]%
        {lincoln1989synergy}
\bibfield{author}{\bibinfo{person}{William Lincoln} {and} \bibinfo{person}{Josef Skrzypek}.} \bibinfo{year}{1989}\natexlab{}.
\newblock \showarticletitle{Synergy of clustering multiple back propagation networks}.
\newblock \bibinfo{journal}{\emph{Advances in neural information processing systems}}  \bibinfo{volume}{2} (\bibinfo{year}{1989}).
\newblock


\bibitem[Long et~al\mbox{.}(2024)]%
        {long2024individual}
\bibfield{author}{\bibinfo{person}{Carol Long}, \bibinfo{person}{Hsiang Hsu}, \bibinfo{person}{Wael Alghamdi}, {and} \bibinfo{person}{Flavio Calmon}.} \bibinfo{year}{2024}\natexlab{}.
\newblock \showarticletitle{Individual arbitrariness and group fairness}.
\newblock \bibinfo{journal}{\emph{Advances in Neural Information Processing Systems}}  \bibinfo{volume}{36} (\bibinfo{year}{2024}).
\newblock


\bibitem[Marx et~al\mbox{.}(2020)]%
        {marx2020predictive}
\bibfield{author}{\bibinfo{person}{Charles Marx}, \bibinfo{person}{Flavio Calmon}, {and} \bibinfo{person}{Berk Ustun}.} \bibinfo{year}{2020}\natexlab{}.
\newblock \showarticletitle{Predictive multiplicity in classification}. In \bibinfo{booktitle}{\emph{International Conference on Machine Learning}}. PMLR, \bibinfo{pages}{6765--6774}.
\newblock


\bibitem[Meyer et~al\mbox{.}(2023)]%
        {meyer2023dataset}
\bibfield{author}{\bibinfo{person}{Anna~P Meyer}, \bibinfo{person}{Aws Albarghouthi}, {and} \bibinfo{person}{Loris D'Antoni}.} \bibinfo{year}{2023}\natexlab{}.
\newblock \showarticletitle{The dataset multiplicity problem: How unreliable data impacts predictions}. In \bibinfo{booktitle}{\emph{Proceedings of the 2023 ACM Conference on Fairness, Accountability, and Transparency}}. \bibinfo{pages}{193--204}.
\newblock


\bibitem[Nie and Wager(2021)]%
        {nie2021quasi}
\bibfield{author}{\bibinfo{person}{Xinkun Nie} {and} \bibinfo{person}{Stefan Wager}.} \bibinfo{year}{2021}\natexlab{}.
\newblock \showarticletitle{Quasi-oracle estimation of heterogeneous treatment effects}.
\newblock \bibinfo{journal}{\emph{Biometrika}} \bibinfo{volume}{108}, \bibinfo{number}{2} (\bibinfo{year}{2021}), \bibinfo{pages}{299--319}.
\newblock


\bibitem[Nosek et~al\mbox{.}(2015)]%
        {nosek2015promoting}
\bibfield{author}{\bibinfo{person}{Brian~A Nosek}, \bibinfo{person}{George Alter}, \bibinfo{person}{George~C Banks}, \bibinfo{person}{Denny Borsboom}, \bibinfo{person}{Sara~D Bowman}, \bibinfo{person}{Steven~J Breckler}, \bibinfo{person}{Stuart Buck}, \bibinfo{person}{Christopher~D Chambers}, \bibinfo{person}{Gilbert Chin}, \bibinfo{person}{Garret Christensen}, {et~al\mbox{.}}} \bibinfo{year}{2015}\natexlab{}.
\newblock \showarticletitle{Promoting an open research culture}.
\newblock \bibinfo{journal}{\emph{Science}} \bibinfo{volume}{348}, \bibinfo{number}{6242} (\bibinfo{year}{2015}), \bibinfo{pages}{1422--1425}.
\newblock


\bibitem[Pedregosa et~al\mbox{.}(2011)]%
        {scikit-learn}
\bibfield{author}{\bibinfo{person}{F. Pedregosa}, \bibinfo{person}{G. Varoquaux}, \bibinfo{person}{A. Gramfort}, \bibinfo{person}{V. Michel}, \bibinfo{person}{B. Thirion}, \bibinfo{person}{O. Grisel}, \bibinfo{person}{M. Blondel}, \bibinfo{person}{P. Prettenhofer}, \bibinfo{person}{R. Weiss}, \bibinfo{person}{V. Dubourg}, \bibinfo{person}{J. Vanderplas}, \bibinfo{person}{A. Passos}, \bibinfo{person}{D. Cournapeau}, \bibinfo{person}{M. Brucher}, \bibinfo{person}{M. Perrot}, {and} \bibinfo{person}{E. Duchesnay}.} \bibinfo{year}{2011}\natexlab{}.
\newblock \showarticletitle{Scikit-learn: Machine Learning in {P}ython}.
\newblock \bibinfo{journal}{\emph{Journal of Machine Learning Research}}  \bibinfo{volume}{12} (\bibinfo{year}{2011}), \bibinfo{pages}{2825--2830}.
\newblock


\bibitem[Redmond(2009)]%
        {misc_communities_and_crime_183}
\bibfield{author}{\bibinfo{person}{Michael Redmond}.} \bibinfo{year}{2009}\natexlab{}.
\newblock \bibinfo{title}{{Communities and Crime}}.
\newblock \bibinfo{howpublished}{UCI Machine Learning Repository}.
\newblock
\newblock
\shownote{{DOI}: https://doi.org/10.24432/C53W3X}.


\bibitem[Roth et~al\mbox{.}(2023)]%
        {Roth2022-sd}
\bibfield{author}{\bibinfo{person}{Aaron Roth}, \bibinfo{person}{Alexander Tolbert}, {and} \bibinfo{person}{Scott Weinstein}.} \bibinfo{year}{2023}\natexlab{}.
\newblock \showarticletitle{Reconciling Individual Probability Forecasts}. In \bibinfo{booktitle}{\emph{Proceedings of the 2023 ACM Conference on Fairness, Accountability, and Transparency}}. \bibinfo{pages}{101--110}.
\newblock


\bibitem[Steyerberg et~al\mbox{.}(2005)]%
        {steyerberg2005equally}
\bibfield{author}{\bibinfo{person}{Ewout~W Steyerberg}, \bibinfo{person}{Marinus~JC Eijkemans}, \bibinfo{person}{Eric Boersma}, {and} \bibinfo{person}{JDF Habbema}.} \bibinfo{year}{2005}\natexlab{}.
\newblock \showarticletitle{Equally valid models gave divergent predictions for mortality in acute myocardial infarction patients in a comparison of logical regression models}.
\newblock \bibinfo{journal}{\emph{Journal of clinical epidemiology}} \bibinfo{volume}{58}, \bibinfo{number}{4} (\bibinfo{year}{2005}), \bibinfo{pages}{383--390}.
\newblock


\bibitem[Tukey(1949)]%
        {tukey1949comparing}
\bibfield{author}{\bibinfo{person}{John~W Tukey}.} \bibinfo{year}{1949}\natexlab{}.
\newblock \showarticletitle{Comparing individual means in the analysis of variance}.
\newblock \bibinfo{journal}{\emph{Biometrics}} (\bibinfo{year}{1949}), \bibinfo{pages}{99--114}.
\newblock


\bibitem[Watson-Daniels et~al\mbox{.}(2023a)]%
        {watson2023multi}
\bibfield{author}{\bibinfo{person}{Jamelle Watson-Daniels}, \bibinfo{person}{Solon Barocas}, \bibinfo{person}{Jake~M Hofman}, {and} \bibinfo{person}{Alexandra Chouldechova}.} \bibinfo{year}{2023}\natexlab{a}.
\newblock \showarticletitle{Multi-target multiplicity: Flexibility and fairness in target specification under resource constraints}. In \bibinfo{booktitle}{\emph{Proceedings of the 2023 ACM Conference on Fairness, Accountability, and Transparency}}. \bibinfo{pages}{297--311}.
\newblock


\bibitem[Watson-Daniels et~al\mbox{.}(2024)]%
        {watson2024predictive}
\bibfield{author}{\bibinfo{person}{Jamelle Watson-Daniels}, \bibinfo{person}{Flavio du~Pin Calmon}, \bibinfo{person}{Alexander D'Amour}, \bibinfo{person}{Carol Long}, \bibinfo{person}{David~C Parkes}, {and} \bibinfo{person}{Berk Ustun}.} \bibinfo{year}{2024}\natexlab{}.
\newblock \showarticletitle{Predictive Churn with the Set of Good Models}.
\newblock \bibinfo{journal}{\emph{arXiv preprint arXiv:2402.07745}} (\bibinfo{year}{2024}).
\newblock


\bibitem[Watson-Daniels et~al\mbox{.}(2023b)]%
        {watson2023predictive}
\bibfield{author}{\bibinfo{person}{Jamelle Watson-Daniels}, \bibinfo{person}{David~C Parkes}, {and} \bibinfo{person}{Berk Ustun}.} \bibinfo{year}{2023}\natexlab{b}.
\newblock \showarticletitle{Predictive multiplicity in probabilistic classification}. In \bibinfo{booktitle}{\emph{Proceedings of the AAAI Conference on Artificial Intelligence}}, Vol.~\bibinfo{volume}{37}. \bibinfo{pages}{10306--10314}.
\newblock


\bibitem[Xin et~al\mbox{.}(2022)]%
        {xin2022exploring}
\bibfield{author}{\bibinfo{person}{Rui Xin}, \bibinfo{person}{Chudi Zhong}, \bibinfo{person}{Zhi Chen}, \bibinfo{person}{Takuya Takagi}, \bibinfo{person}{Margo Seltzer}, {and} \bibinfo{person}{Cynthia Rudin}.} \bibinfo{year}{2022}\natexlab{}.
\newblock \showarticletitle{Exploring the whole rashomon set of sparse decision trees}.
\newblock \bibinfo{journal}{\emph{Advances in neural information processing systems}}  \bibinfo{volume}{35} (\bibinfo{year}{2022}), \bibinfo{pages}{14071--14084}.
\newblock


\bibitem[Zhong et~al\mbox{.}(2024)]%
        {zhong2024exploring}
\bibfield{author}{\bibinfo{person}{Chudi Zhong}, \bibinfo{person}{Zhi Chen}, \bibinfo{person}{Jiachang Liu}, \bibinfo{person}{Margo Seltzer}, {and} \bibinfo{person}{Cynthia Rudin}.} \bibinfo{year}{2024}\natexlab{}.
\newblock \showarticletitle{Exploring and interacting with the set of good sparse generalized additive models}.
\newblock \bibinfo{journal}{\emph{Advances in neural information processing systems}}  \bibinfo{volume}{36} (\bibinfo{year}{2024}).
\newblock


\end{thebibliography}

\newpage

\appendix

\section{Full Reconcile Algorithm}\label{appendix-sec:reconcilealgo}

Recently, \cite{Roth2022-sd} made a significant step towards resolving the problem of predictive multiplicity: they showed that, while individual probabilities are inherently unknowable, they are \emph{falsifiable}. 
Namely, given two different models $f_1$, $f_2$ that are learned on the same data, we can perform an efficient reconciliation on $(f_1, f_2)$.
If $f_1, f_2$ already agree in their predictions, then we are done.
Otherwise, we can efficiently define a group (which identifies a subset of the individuals) that ``witnesses'' this disagreement.
The quantity we are interested in quantifying is the following:
\begin{definition}[$\alpha$-approx. group cond. mean consistency]\label{definition:meanconsistency}
    Given a model $f: \X \rightarrow [0,1]$, a group $g: \X \rightarrow \{0,1\}$, and a parameter $\alpha>0$, we say that $f$ satisfies $\alpha$-\emph{approximate group conditional mean consistency} with respect to $g$ if
    \[
        \Big( \E_{(x, y) \sim \D}[f(x) | g(x)=1] - \E_{(x, y) \sim \D[y |g(x)=1]} \Big)^2 \leq \dfrac{\alpha}{\mu(g)}.
    \]
\end{definition}
A group $g$ is a witness of disagreement if it witnesses a violation of $\alpha$-approximate group conditional mean consistency.

Then, we can use this group to update either $f_1$ or $f_2$, so that the updated model now makes correct predictions on average with respect to this fixed group. 
Suppose that $f'_1$ is the model that has been patched through this update. 
Next, their algorithm iteratively repeats the following step: if $f'_1$ and $f_2$ do not present significant disagreement between them, we terminate.
Otherwise, we find a group that witnesses this disagreement and use it to patch one of the two models.
We continue this procedure until the two models no longer significantly disagree; we call $(f_1^T, f_2^T)$ the final two models.\footnote{While the Reconcile algorithm technically returns two updated models, rather than one merged model, the two models agree on most of their predictions. In some cases, as we discuss later, it will be useful to view the output of Reconcile as a single model.} 

In \citep{Roth2022-sd}, the authors then show three key properties of their Reconcile algorithm: (1) The iterative process converges quickly.
(2) At every step $0 < t \leq T$ of the algorithm, the model $f^t_1$ is better than $f^{t-1}_1$, and likewise for $f_2$. 
By ``better'' we mean more accurate: namely, each step of the Reconcile process reduces the square loss of the model compared to the true labels that we obtain from the ground-truth model $f^*$.
(3) When the Reconcile algorithm terminates, the models $(f_1^T, f_2^T)$ do not present significant disagreement (which we quantify with a parameter $\alpha>0$).
We remark that the Reconcile algorithm is heavily inspired by the recent literature on multigroup fairness, particularly on the development of the notion of \emph{multicalibration} \citep{hebert2018multicalibration, gopalan2023characterizing}.
As we further discuss in the appendix, the Reconcile algorithm closely resembles the multiaccuracy algorithm: in both cases, we find a group that witnesses a violation of the desired property that we have defined, and then we use this group to patch the current predictor and make progress towards $f^*$ measured in square loss. 

In sum, the key insight from \citep{Roth2022-sd} is that, whenever the predictive multiplicity problem arises, we can resolve it through an efficient reconciliation process that builds a more accurate model. 
An important difference between the Reconcile approach and some of the model multiplicity literature is that some previous works (e.g., those that deal with Rashomon sets \citep{hsu2022rashomon})
restrict the models to lie within a fixed hypothesis class $\Ha$.
In contrast, the updated models produced during the Reconcile algorithm do not necessarily belong to the same model class as $f_1$ or $f_2$.
Therefore, whenever there is no need to remain within the hypothesis class $\Ha$, the Reconcile algorithm constitutes a provable solution, with good theoretical guarantees, to the predictive multiplicity problem.

We include the full pseudocode of the Reconcile algorithm taken from \cite{Roth2022-sd} for completeness.

\begin{algorithm}[H]
\caption{\textsc{Reconcile}($f_1,f_2,\alpha,\epsilon,\D$) \cite{Roth2022-sd}}
\label{alg:reconciler}
\begin{algorithmic}[1]  

\STATE \textbf{Initialize:} 
\[
  t \;=\; t_1 \;=\; t_2 \;=\; 0, 
  \quad
  f_1^{t_1} = f_1,
  \quad
  f_2^{t_2} = f_2.
\]

\STATE \(\displaystyle
  m = \left\lceil \frac{2}{\sqrt{\alpha}\,\epsilon} \right\rceil.
\)

\WHILE{\(\mu\bigl(U_\epsilon(f_1^{t_1},\,f_2^{t_2})\bigr) \;\ge\; \alpha\)}
    \STATE \textbf{For each} \(\bullet \in \{>,<\}\) \textbf{and} \(i \in \{1,2\}\), define:
    \[
      v_*^\bullet 
      \;=\; 
      \mathbb{E}_{(x,y)\sim \D}\!\bigl[y \mid x \in U_\epsilon^\bullet(f_1^{t_1},\,f_2^{t_2})\bigr],
      \quad
      v_i^\bullet 
      \;=\; 
      \mathbb{E}_{(x,y)\sim \D}\!\bigl[f_i^{t_i}(x) \mid x \in U_\epsilon^\bullet(f_1^{t_1},\,f_2^{t_2})\bigr].
    \]

    \STATE Let
    \[
      (i_t,\bullet_t) 
      \;=\;
      \underset{\substack{i \in \{1,2\}\\ \bullet \in \{>,<\}}}{\arg\max}
      \bigl[\,
        \mu\bigl(U_\epsilon^\bullet(f_1^{t_1},\,f_2^{t_2})\bigr)
        \cdot
        (\,v_*^\bullet - v_i^\bullet\,)^2
      \bigr]
    \]
    breaking ties arbitrarily.

    \STATE Define
    \[
      g_t(x) 
      \;=\; 
      \begin{cases}
      1, & \text{if } x \in U_\epsilon^{\bullet_t}(f_1^{t_1},\,f_2^{t_2}), \\[4pt]
      0, & \text{otherwise}.
      \end{cases}
    \]

    \STATE \(\displaystyle
      \tilde{\Delta}_t 
      = 
      \mathbb{E}_{(x,y)\sim \D}\!\bigl[y \mid g_t(x) = 1\bigr]
      \;-\;
      \mathbb{E}_{(x,y)\sim \D}\!\bigl[f_{i_t}^{\,t_{i_t}}(x) \mid g_t(x)=1\bigr],
    \quad
      \Delta_t 
      =
      \text{Round}(\tilde{\Delta}_t;\,m).
    \)

    \STATE \(\displaystyle
      \tilde{f}_i^{\,t_i+1}(x) 
      = 
      h\bigl(x,\,f_i^{t_i},\,g_t,\,\Delta_t\bigr),
      \quad
      f_i^{\,t_i+1}(x) 
      = 
      \text{Project}\bigl(\tilde{f}_i^{t_i+1},\,[0,1]\bigr),
    \quad
      t_i \;=\; t_i + 1,
      \quad
      t \;=\; t + 1.
    \)
\ENDWHILE

\STATE \textbf{Output:} \((f_1^{t_1},\,f_2^{t_2})\).

\end{algorithmic}
\end{algorithm}

\section{Extending Reconcile to CI}\label{sec:CI-extension}
\subsection{Application of the Reconcile Algorithm to CATE}

Here we provide a full, self-contained description of the ReconcileCATE algorithm, without reducing it directly to Reconcile.
While the algorithm and proof are an adaptation of the original algorithm to the domain of causal inference, and hence follows exactly the same structure,
we include it as to have a standalone proof and full description of the ReconcileCATE algorithm that does not rely on a black-box call to Reconcile, which we hope can be helpful to practitioners who are only interested in the causal setting.

All notation and structure follows exactly the original description of the Reconcile algorithm by \citep{Roth2022-sd} and we do not claim its originality.
We hope the same use of notation can ease the understanding of the two papers together.

\begin{definition}
    Two models $\hat{\tau}_1$ and $\hat{\tau}_2$ have an $\epsilon$-disagreement on a point $x \in X$ if $|\hat{\tau}_1(x) - \hat{\tau}_2(x)| > \epsilon$.
Let $U_\epsilon(\hat{\tau}_1, \hat{\tau}_2)$ be the set of points on which $\hat{\tau}_1$ and $\hat{\tau}_2$ $\epsilon$-disagree:
$$U_\epsilon(\hat{\tau}_1, \hat{\tau}_2) = \{x : |\hat{\tau}_1(x) - \hat{\tau}_2(x)| > \epsilon\}$$

\end{definition}

We can use \textit{subset groups} to falsify models by demonstrating that models significantly disagree on their average treatment effect for that group. If so, the model with the lower group accuracy must be incorrect. We can model groups $g$ as indicator functions $g: \mathcal{X} \rightarrow \{0,1\}$ where our distribution on inputs $x$ follows from a joint distribution $D$.

\begin{definition}
    Under a distribution $\mathcal{D}$, a group $g:X \rightarrow \{0,1\}$ has probability mass $\mu(g)$ defined as:
$$\mu(g) = \underset{(x,y,t) \sim \mathcal{D}}{\Pr}[g(x) =1].$$
\end{definition}

For a model $\hat{\tau}$ and a group $g$, we define \textit{approximate group conditional mean consistency} to capture the extent to which the average prediction of the model on points in $g$ compares to the expected outcome on points in $g$. Models $\hat{\tau}$ are falsified when they do not satisfy approximate group conditional mean consistency on any group $g$.

\begin{definition}
    For those in group $g$, $\hat{\tau}$ satisfies $\alpha$-approximate group conditional mean consistency with respect to a covariate $X$ if

\begin{align*}
    &\Bigg(\underset{(x, y, t) \sim \mathcal{D}}{\mathbb{E}}[\hat{\tau}(x)|g(x)=1] \\
    &- \left(\underset{(x, y, t) \sim \mathcal{D}}{\mathbb{E}}[y|g(x)=1, T=1] - \underset{(x, y, t) \sim \mathcal{D}}{\mathbb{E}}[y|g(x)=1, T=0] \right) \Bigg)^2 \leq \frac{\alpha}{\mu(g)} .   
\end{align*}

\end{definition}

This is analogous to comparing their squared error. We normalize with $\mu(g)$ to account for varying group size. 

\begin{definition}
    Fix any two models $\hat{\tau}_1, \hat{\tau}_2: \mathcal{X} \mapsto [0,1]$ and any $\epsilon > 0$. Define the sets: 

    \begin{align*}
        U_{\epsilon}^>(\hat{\tau}_1, \hat{\tau}_2) &= \{x \in U_{\epsilon}(\hat{\tau}_1, \hat{\tau}_2): \hat{\tau}_1(x)> \hat{\tau}_2\}\\
        U_{\epsilon}^<(\hat{\tau}_1, \hat{\tau}_2) &= \{x \in U_{\epsilon}(\hat{\tau}_1, \hat{\tau}_2): \hat{\tau}_1(x)< \hat{\tau}_2\}
    \end{align*}

    Based on these sets for $\bullet \in \{>, <\}$ and $i \in \{1,2\}$, 
    \begin{align*}
        v_*^\bullet &= \underset{(x,y,t \sim D)}{\E}[y|x\in U_{\epsilon}^\bullet(\hat{\tau}_1, \hat{\tau}_2), T=1] - \underset{(x,y,t \sim D)}{\E}[y|x\in U_{\epsilon}^\bullet(\hat{\tau}_1, \hat{\tau}_2), T=0],\\
        v_i^\bullet &= \underset{(x,y,t \sim D)}{\E} [\hat{\tau}_i(x)|x \in U_{\epsilon}^\bullet(\hat{\tau}_1, \hat{\tau}_2)].
    \end{align*}
\end{definition}

\begin{lemma}
    For models $\tau_1$ and $\tau_2$ with an $\epsilon$-disagreement mass $\mu(U_{\epsilon}(\hat{\tau}_1, \hat{\tau}_2)) = \alpha$, 

    \begin{align*}
        \mu(U_{\epsilon}^\bullet(\hat{\tau}_1, \hat{\tau}_2)) \cdot (v_*^\bullet-v_i^\bullet)^2 \geq \frac{\alpha \epsilon^2}{8}.
    \end{align*}
\end{lemma}

\begin{proof}
    We show that two inequalities hold. 
    First, we show that $\mu(U_{\epsilon}^\bullet)(\hat{\tau}_1, \hat{\tau}_2) \geq \frac{\alpha}{2}$. We know this is true because $U_{\epsilon}(\hat{\tau}_1, \hat{\tau}_2)$ can be written as the disjoint union 

    $$U_{\epsilon}(\hat{\tau}_1, \hat{\tau}_2) = U_{\epsilon}^>(\hat{\tau}_1, \hat{\tau}_2) \cup U_{\epsilon}^<(\hat{\tau}_1, \hat{\tau}_2).$$

    Next, we show that $|v_i^\bullet - v_*^\bullet| \geq \frac{\epsilon}{2}$ for at least one $\bullet \in \{>, <\}$. This follows from the assumption that $U_{\epsilon}^\bullet(\hat{\tau}_1, \hat{\tau}_2)$ are $\epsilon$-separated, so $|v_i^\bullet - v_*^\bullet| \geq \epsilon$ and one model's accuracy is lower-bounded by $\frac{\epsilon}{2}$. 

    Taken together, the two inequalities combine to yield 
    $$ \mu(U_{\epsilon}^\bullet(\hat{\tau}_1, \hat{\tau}_2)) \cdot (v_*^\bullet-v_i^\bullet)^2 \geq \frac{\alpha \epsilon^2}{8}.$$
\end{proof}

On some subset of the data for which $\tau_1$ and $\tau_2$ disagree in a congruent direction, their scaled error is at least $\frac{\alpha \epsilon^2}{8}$. 

\begin{lemma}
    Fix any model $\hat{\tau}_t: \mathcal{X} \mapsto [0,1]$, group $g_t: \mathcal{X} \mapsto \{0,1\}$, and distribution of $D$. 

    Let 
    \begin{align*}
        \Delta_t &=\left(\underset{(x, y, t) \sim \mathcal{D}}{\mathbb{E}}[y|g(x)=1, T=1] - \underset{(x, y, t) \sim \mathcal{D}}{\mathbb{E}}[y|g(x)=1, T=0] \right)\\
    & - \underset{(x, y, t) \sim \mathcal{D}}{\mathbb{E}}[\hat{\tau}(x)|g(x)=1]
    \end{align*}

    We define our new model $\hat{\tau}_{t+1} = h(x, \hat{\tau}_t;g_t, \Delta_t)$ where $h$ is a "patch" defined as

    \[
    h(x, \hat{\tau}_t;g_t, \Delta_t) = \begin{cases}
        \hat{\tau}(x) + \Delta & g(x)=1\\
        \hat{\tau}(x) & \text{otherwise}
    \end{cases}
    \]

    Then $$B(\hat{\tau}_t, D) - B(\hat{\tau}_{t+1}, D) = \mu(g_t)\cdot \Delta_t^2.$$
\end{lemma}

Given a model $\tau_t$ and group $g_t$ that witness $\alpha$-approximate group conditional mean consistency on $\tau_t$, we can produce a $\tau_{t+1}$ with smaller  Brier score by exactly $\alpha$. 

\begin{proof}
    All expectations will be taken over $(x, y, t) \sim \mathcal{D}$. 

    \begin{align*}
        B(\hat{\tau}_t, D) &- B(\hat{\tau}_{t+1}, D) \\ &= \Pr[g_t(x)=0] \cdot \\
        &\big[\big(\E[\hat{\tau}_{t}(x)|g_{t}(x)=0] - \left(\E[y|g_{t}(x)=0, T=1] - \E[y|g_{t}(x)=0, T=0] \right) \big)^2 \\
    &- \big(\E[\hat{\tau}_{t+1}(x)|g_{t}(x)=0] - \left(\E[y|g_{t}(x)=0, T=1] - \E[y|g_{t}(x)=0, T=0] \right) \big)^2 \big]\\
    &+ \Pr[g_t(x)=1]\cdot \\
    &\big[\big(\E[\hat{\tau}_{t}(x)|g_{t}(x)=1] - \left(\E[y|g_{t}(x)=1, T=1] - \E[y|g_{t}(x)=1, T=0] \right) \big)^2 \\
    &- \big(\E[\hat{\tau}_{t+1}(x)|g_{t}(x)=1] - \left(\E[y|g_{t}(x)=1, T=1] - \E[y|g_{t}(x)=1, T=0] \right) \big)^2 \big].\\
    \end{align*}

    Our first term equals 0 because $\hat{\tau}_t$ and $\hat{\tau}_{t+1}$ only differ in predictions for $x$ such that $g(x)=1$, so the previous expression is equal to

    \begin{align*}
        &= 0 + \Pr[g_t(x)=1]\cdot \\
    &\big[\big(\E[\hat{\tau}_{t}(x)|g_{t}(x)=1] - \left(\E[y|g_{t}(x)=1, T=1] - \E[y|g_{t}(x)=1, T=0] \right) \big)^2 \\
    &- \big(\E[\hat{\tau}_{t+1}(x)|g_{t}(x)=1] - \left(\E[y|g_{t}(x)=1, T=1] - \E[y|g_{t}(x)=1, T=0] \right) \big)^2 \big]\\
    &= \mu(g)\left[ 2\Delta_t \left[ \E[y|g_t(x)=1, T=1] - \E[y|g_t(x)=1, T=0] - \E[\hat{\tau}_t|g_t(x)=1] - \Delta_t^2\right] \right]\\
    &= \mu(g) \left[2 \Delta_t^2 - \Delta_t^2\right] = \mu(g) \Delta_t^2.
    \end{align*}
\end{proof}

Therefore, when two models $\epsilon$-disagree on an $\alpha$-fraction of points, we can \textit{constructively falsify} at least one model and update its squared error by $O(\alpha \epsilon^2)$.

Next, we define a function \textit{Round}. 

\begin{definition}
    For any integer $m$, $[1/m]$ is the set of $m+1$ grid points: 

    $$\left[ \frac{1}{m}\right] = \left\{ 0, \frac{1}{m}, \cdots, \frac{m-1}{m}, 1 \right\}$$

    For any value $v \in [0,1]$, define $Round(v;m) = \underset{v' \in [1/m]}{\arg\min} |v-v'|$ return the closest grid point to $v$ in $[1/m]$.  
\end{definition}

Note that for $v'=Round(v;m)$ always returns a $v'$ such that $|v-v'| \leq \frac{1}{2m}$. 

\begin{algorithm}[H]
\caption{\textsc{ReconcileCATE}$\bigl(\hat{\tau}_1,\hat{\tau}_2,\alpha,\epsilon,\mathcal{D}\bigr)$) 
         \quad(\emph{A procedure to reconcile two CATE estimators})}
\label{alg:ReconcileCATE}
\begin{algorithmic}[1]

\STATE \textbf{Initialize:}
\[
  t \;=\; t_1 \;=\; t_2 \;=\; 0,
  \quad
  \hat{\tau}_1^{t_1} \;=\; \hat{\tau}_1,
  \quad
  \hat{\tau}_2^{t_2} \;=\; \hat{\tau}_2.
\]

\STATE \(\displaystyle
  m 
  \;=\;
  \left\lceil \frac{2}{\sqrt{\alpha}\,\epsilon} \right\rceil.
\)

\WHILE{\(\mu\bigl(U_\epsilon(\hat{\tau}_1^{t_1},\,\hat{\tau}_2^{t_2})\bigr)
        \;\ge\;\alpha\)}

    \STATE \textbf{For each} \(\bullet \in \{>,<\}\) \textbf{and} \(i \in \{1,2\}\), define:
    \[
      v_*^\bullet 
      \;=\; 
      \mathbb{E}_{(x,y,T)\sim \mathcal{D}}
        \Bigl[y \,\Big|\,
              x \in U_\epsilon^\bullet(\hat{\tau}_1^{t_1},\,\hat{\tau}_2^{t_2}),\; T=1\Bigr]
      \;-\;
      \mathbb{E}_{(x,y,T)\sim \mathcal{D}}
        \Bigl[y \,\Big|\,
              x \in U_\epsilon^\bullet(\hat{\tau}_1^{t_1},\,\hat{\tau}_2^{t_2}),\; T=0\Bigr],
    \]
    \[
      v_i^\bullet 
      \;=\; 
      \mathbb{E}_{(x,y,T)\sim \mathcal{D}}
        \Bigl[\hat{\tau}_i^{\,t_i}(x)
        \,\Big|\,
        x \in U_\epsilon^\bullet(\hat{\tau}_1^{t_1},\,\hat{\tau}_2^{t_2})\Bigr].
    \]

    \STATE Let
    \[
      (i_t,\,\bullet_t)
      \;=\;
      \underset{\substack{i \in \{1,2\}\\ \bullet \in \{>,<\}}}{\arg\max}
      \Bigl[\,
        \mu\bigl(U_\epsilon^\bullet(\hat{\tau}_1^{t_1},\,\hat{\tau}_2^{t_2})\bigr)
        \;\cdot\;
        \bigl(v_*^\bullet \;-\; v_i^\bullet\bigr)^2
      \Bigr]
      \quad
      \text{(breaking ties arbitrarily)}.
    \]

    \STATE Define
    \[
      g_t(x) 
      \;=\;
      \begin{cases}
        1, & \text{if } x \in U_\epsilon^{\bullet_t}(\hat{\tau}_1^{t_1},\,\hat{\tau}_2^{t_2}), \\[3pt]
        0, & \text{otherwise}.
      \end{cases}
    \]

    \STATE \(\displaystyle
      \tilde{\Delta}_t
      \;=\;
      \Bigl(
         \mathbb{E}_{(x,y,T)\sim \mathcal{D}}
           \bigl[y \mid g_t(x)=1,\;T=1\bigr]
       \;-\;
         \mathbb{E}_{(x,y,T)\sim \mathcal{D}}
           \bigl[y \mid g_t(x)=1,\;T=0\bigr]
      \Bigr)
      \;-\;
      \mathbb{E}_{(x,y,T)\sim \mathcal{D}}
        \bigl[\hat{\tau}_{i_t}^{\,t_{i_t}}(x)\mid g_t(x)=1\bigr],
    \)
    \[
      \Delta_t 
      \;=\;
      \text{Round}\bigl(\tilde{\Delta}_t,\;m\bigr).
    \]

    \STATE \(\displaystyle
      \hat{\tau}_{i_t}^{\,t_{i_t}+1}(x)
      \;=\;
      h\bigl(x,\;\hat{\tau}_{i_t}^{\,t_{i_t}},\;g_t,\;\Delta_t\bigr),
      \quad
      t_{i_t}
      \;=\;
      t_{i_t} \;+\; 1,
      \quad
      t
      \;=\;
      t \;+\; 1.
    \)

\ENDWHILE

\STATE \textbf{Output:} 
\(
  \bigl(\hat{\tau}_1^{t_1},\;\hat{\tau}_2^{t_2}\bigr).
\)

\end{algorithmic}
\end{algorithm}

\begin{theorem}
    Given any pair of models $\hat{\tau}_1, \hat{\tau}_2: \mathcal{X} \mapsto [0,1]$, any distribution $\mathcal{D}$, and any $\alpha, \epsilon \geq 0$, Algorithm 1 executes $T=T_1+T_2$ many rounds and returns a pair of models $(\hat{\tau}_1^{T_1},\hat{\tau}_1^{T_2})$ such that the following inequalities hold: 

    \begin{enumerate}
        \item $T \leq (B(\hat{\tau}_1, \mathcal{D})+B(\hat{\tau}_2, \mathcal{D})) \cdot \frac{16}{\alpha \epsilon^2}$
        \item $B(\hat{\tau}_1^{T_1}, \mathcal{D}) \geq B(\hat{\tau}_1, D) - T_1 \cdot \frac{\alpha \epsilon^2}{16} \text{ and } B(\hat{\tau}_2^{T_2}, \mathcal{D}) \geq B(\hat{\tau}_2, D) - T_2 \cdot \frac{\alpha \epsilon^2}{16}$
        \item $\mu(U_\epsilon (\hat{\tau}_1^{T_1}, \hat{\tau}_2^{T_2}) \leq \alpha$
    \end{enumerate}
\end{theorem}

The first inequality indicates that the initial performance of each model on the dataset limits the total number of rounds. This implies that the reconciliation procedure converges quickly. The second inequality specifies that the improvement in accuracy for $\hat{\tau}_1$ from its initial state to its final state $\hat{\tau}_1^{T_1}$ is at least $T_1 \cdot \frac{\alpha \epsilon^2}{16}$. This means that each model can improve its accuracy by a maximum of $\frac{\alpha \epsilon^2}{16}$ in each iteration of the algorithm where it is updated. A similar inequality holds for $\hat{\tau}_2$, indicating that both output models are more accurate than their initial versions. The third inequality asserts that the final $\epsilon$-disagreement mass between the two models is no more than $\alpha$, suggesting that the final models $(\hat{\tau}_1^{T_1}, \hat{\tau}_2^{T_2})$ almost always agree on their predictions of individual probabilities. Additionally, every intermediate model $\hat{\tau}_1^{t_1}$ and $\hat{\tau}_2^{t_2}$ that was considered but not selected as the final output was rejected because it did not meet the required level of $\alpha$-approximate group conditional mean consistency for the specified $\alpha$.

\begin{proof}
    From Lemma 1, we have that 
    $$\mu(U_{\epsilon}^{\bullet_t}(\hat{\tau}_1, \hat{\tau}_2)) \cdot (v_*^\bullet-v_i^\bullet)^2 \geq \frac{\alpha \epsilon^2}{8}.$$

    Now, let $\tilde{\hat{\tau}}_t^{t_i+1}=h(x,\hat{\tau}_i^{t_i}, g_t, \tilde{\Delta}_t)$. Recall that $\tilde{\Delta}_t$ is the rounded version of $\Delta_t$, so $\tilde{\hat{\tau}}_t^{t_i+1}$ is the update to $\hat{\tau}_t^{t_i}$ that would have resulted from using the unrounded measurement $\tilde{\Delta}_t$.

    Applying Lemma 2, we conclude $$B(\hat{\tau}_t^{t_i}, \mathcal{D})-B(\tilde{\hat{\tau}}_t^{t_i+1}, \mathcal{D}) \geq \frac{\alpha \epsilon^2}{8}.$$

We can now compute that:

\begin{align*}
(B(\hat{\tau}_t^{t_i}, \mathcal{D}) - B(\hat{\tau}_t^{t_i+1}, \mathcal{D})) &= (B(\hat{\tau}_t^{t_i}, \mathcal{D}) - B(\tilde{\hat{\tau}}_t^{t_i+1}, \mathcal{D}))-(B(\hat{\tau}^{t_i+1}, \mathcal{D}) - B(\tilde{\hat{\tau}}_t^{t_i+1}, \mathcal{D})) \\
&\geq \frac{\alpha \epsilon^2}{8} - (B(\hat{\tau}_t^{t_i+1}, \mathcal{D}) - B(\tilde{\hat{\tau}}_t^{t_i+1}, \mathcal{D})). \\
\end{align*}
Now, we can upper bound the final term, $B(\tilde{\hat{\tau}}_t^{t_i+1}, \mathcal{D})$. Let $\hat{\Delta} = \tilde{\Delta}_t-\Delta_t$, the difference between the rounded and unrounded values. Observe that:
\begin{enumerate}
    \item $\tilde{\hat{\tau}}_t^{t_i+1} = h(x, \hat{\tau}_t^{t_i+1}, g_t, \hat{\Delta})$
    \item \begin{align*}
    \hat{\Delta} &= \left(\underset{(x, y, t) \sim \mathcal{D}}{\mathbb{E}}[y|g_t(x)=1, T=1] - \underset{(x, y, t) \sim \mathcal{D}}{\mathbb{E}}[y|g_t(x)=1, T=0] \right)- \underset{(x, y, t) \sim \mathcal{D}}{\mathbb{E}}[\hat{\tau}_{i}^{t_i}(x)|g_t(x)=1] - \Delta_t\\
    &= \left(\underset{(x, y, t) \sim \mathcal{D}}{\mathbb{E}}[y|g_t(x)=1, T=1] - \underset{(x, y, t) \sim \mathcal{D}}{\mathbb{E}}[y|g_t(x)=1, T=0] \right)- \underset{(x, y, t) \sim \mathcal{D}}{\mathbb{E}}[\hat{\tau}_{i}^{t_i+1}(x)|g_t(x)=1].
        \end{align*}
    \item By definition of the Round operation, \(|\hat{\Delta}| \leq \frac{1}{2m}\).
\end{enumerate}

We apply Lemma 2 to conclude:

\begin{align*}
B(\hat{\tau}_t^{t_i+1}, \mathcal{D}) - B(\tilde{\hat{\tau}}_t^{t_i+1}, \mathcal{D}) = \mu(g_t) \hat{\Delta}^2 \leq \frac{1}{4m^2}.
\end{align*}

Recall that we set $m \geq \frac{2}{\sqrt{\alpha}\epsilon}$. Combine this with our initial calculation to conclude:
$$B(\hat{\tau}_t^{t_i}, \mathcal{D}) - B(\hat{\tau}_t^{t_i+1}, \mathcal{D}) \geq \frac{\alpha \epsilon^2}{8}-\frac{1}{4m^2} \geq \frac{\alpha \epsilon^2}{16}$$

Applying this lemma separately for $T_1$ and $T_2$ updates for $\hat{\tau}_1$ and $\hat{\tau}_2$ models, respectively, yields
$$B(\hat{\tau}_1^{T_1}, \mathcal{D}) \geq B(\hat{\tau}_1, D) - T_1 \cdot \frac{\alpha \epsilon^2}{16} \text{ and } B(\hat{\tau}_2^{T_2}, \mathcal{D}) \geq B(\hat{\tau}_2, D) - T_2 \cdot \frac{\alpha \epsilon^2}{16}.$$

Given that Brier scores are non-negative, we know that both $B(\hat{\tau}_1^{T_1},\mathcal{D}) \geq 0$ and $ B(\hat{\tau}_2^{T_2},\mathcal{D}) \geq 0$, so the above inequalities can be arranged to yield
$$T_1 \leq B(\hat{\tau}_1, \mathcal{D})\cdot \frac{16}{\alpha \epsilon^2} \text{ and } T_2 \leq B(\hat{\tau}_2, \mathcal{D})\cdot \frac{16}{\alpha \epsilon^2}.$$

Taken together, $T = T_1+T_2 \leq (B(\hat{\tau}_1, \mathcal{D})+B(\hat{\tau}_2, \mathcal{D})) \cdot \frac{16}{\alpha \epsilon^2}$

Finally, the halting condition of the algorithm implies that $\mu(U_\epsilon(\hat{\tau}_1^{T_1}, \hat{\tau}_2^{T_2})) \leq \alpha$. 
\end{proof}

Thus, given any two models with substantial disagreement, we can make strictly improved models. This means that we can never have two equally accurate models with substantial disagreements that cannot be improved—because their coexistence is used as a falsification for the minimally achievable loss—in our case, Brier score. 

So, for any pair of models that have been reconciled via the above algorithm, they must produce similar individual probability predictions and agree on almost all inputs. For sufficiently large reference classes, both models are consistent with the data within a certain bound and cannot, by design, significantly disagree.

\section{Reconcile Experiments} \label{appendix: experiments}
All the experiments discussed in Section \ref{sec:experiments} and \ref{sec:CATE experiments} are visualized in Figure \ref{fig:appendix experiments}. These experiments were conducted locally using a system equipped with an M2 chip, featuring an 8-core CPU and a 10-core integrated GPU, along with 16 GB of unified memory.
The full results of all our experiments for Sections \ref{sec:experiments Reconcile results}, \ref{sec:experiments Reconcile and MM} and \ref{sec:experiments solution} are available in the code repository.
\begin{figure*}[!ht]
\centering
\includegraphics[width=\columnwidth]{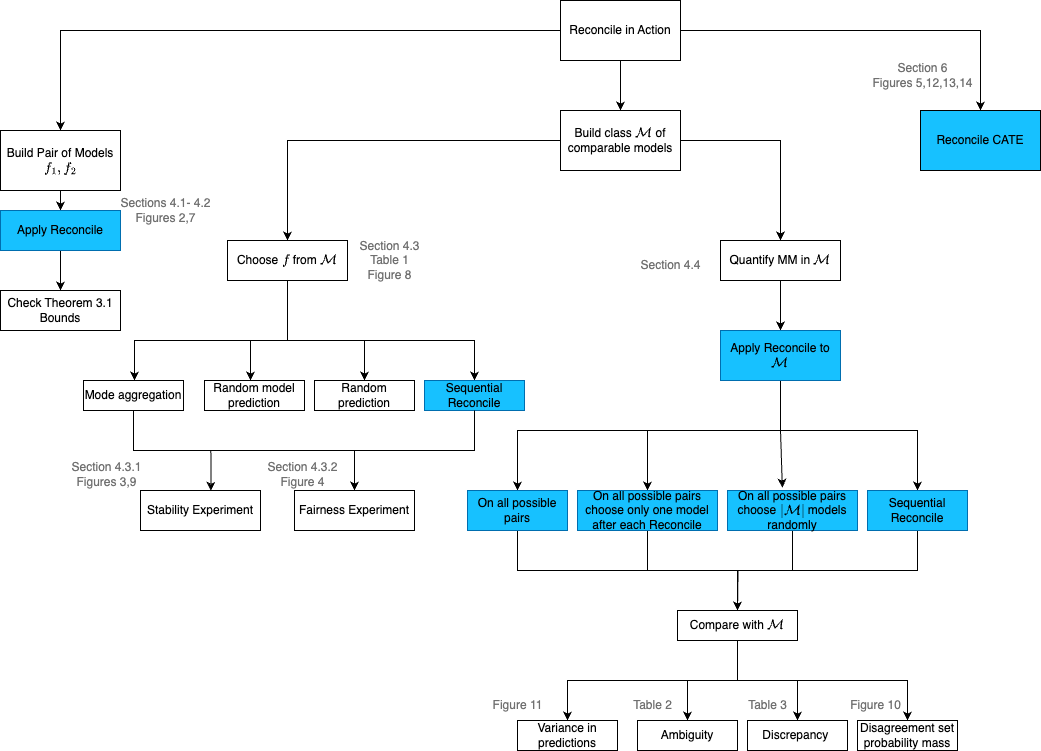}
\caption{Visual representation of goals and experiments in Section \ref{sec:experiments} and \ref{sec:CATE experiments}. All experiments that include applying Reconcile are colored blue. The corresponding sections are written outside the boxes.} \label{fig:appendix experiments}
\end{figure*}

\subsection{Descriptions of Datasets} \label{appendix:datasets}
Below are more details on the 5 datasets used in all our experiments in Section \ref{sec:experiments}.
\begin{enumerate}
    \item \textbf{The UCI Adult dataset}: The Adult dataset \citep{misc_adult_2} consists of 48842 individuals with demographic (e.g., age, race, and gender), education (degree), employment (occupation, hours-per-week), personal (marital status, relationship), and financial (capital gain/loss) features. The task is to predict whether an individual’s income exceeds \$50K per year or not (this corresponds to a binary classification task).
    \item \textbf{The COMPAS dataset}: COMPAS (Correctional Offender Management Profiling for Alternative Sanctions) \citep{Larson2016a} includes $1,0000$ individuals representing defendants
    released on bail. The task is to predict whether to release a defendant on bail or not (this corresponds to a binary classification task) using features, such as criminal history, jail, prison time, and defendant’s demographics. 
    \item  \textbf{The Communities and Crime dataset}: This dataset from UCI \citep{misc_communities_and_crime_183}, which contains $1{,}994$ communities in the United States described by 140 features relevant to crime
(e.g., police per population, income) as well as demographic features (such as race and sex). 
The goal is to predict the violent crime rate in the community (this corresponds to a continuous regression task).
    \item \textbf{Statlog German Credit Dataset}: The German Credit Data \citep{statlog_(german_credit_data)_144} consists of 1,000 borrower records, categorized into two classes: accepted and rejected applicants. Each instance is described by 20 input attributes, which include demographic factors like age, sex, and marital status, along with financial factors such as properties, job, and credit history. For our experiment, we utilize the version of the dataset provided by Strathclyde University, which contains all numeric values.
    \item \textbf{The American Community Survey (ACS)
    Public Use Microdata Sample (PUMS)}: The ACS PUMS dataset \citep{ding2021retiring} covers multiple years and all states across the United States. It supports five different prediction tasks: income, residential address change (mobility), commuting time to workplace, public health insurance, and employment. For our experiments, we focus on the first three tasks, limiting the data to the year 2018 for the state of Florida, yielding $98{,}925$, $34{,}491$, and $88{,}071$ instances for each respective task. The dataset comprises 14 features, including age, race, sex, class of work, and education level. The prediction tasks are whether an individual’s income is above $\$50,000$, whether an individual had the same residential address one year ago, and whether an individual has a commute to work that is longer than 20 minutes for the three prediction tasks that we have chosen.
\end{enumerate}

For our causal inference experiments in Section \ref{sec:CATE experiments}, we use the \textbf{Twins dataset} \citep{almond2005costs, guo2020survey} and National Study dataset \citep{nosek2015promoting}. The Twins dataset includes twins who are the same sex and weigh less than 2 kg. The treatment $T=1$ is being born the heavier twin and the outcome is the mortality of each of the twins in their first year of life. Twins represent counterfactuals which are usually produced synthetically in other benchmarks. 
After pre-processing the data, we end up with $23,968$ instances.

The National Study of Learning Mindsets is a randomized study conducted in U.S. public high schools, aimed at evaluating the impact of a nudge-like intervention designed to foster a growth mindset among students on their academic achievement. We utilize data from experiments described in \cite{athey2019estimating}. This dataset includes information on $10,391$ students from 76 schools, featuring a binary treatment indicator ($Z$), a real-valued outcome $Y$, and 10 covariates that are either categorical or real-valued, such as a student race, gender, and school achievement level.

\subsection{Models} \label{appendix:models}
Below are the details of the models used in the experiments in Sections \ref{sec:experiments Reconcile and MM} and \ref{sec:experiments solution} for building the class $\M$. For the fairness experiment outlined in Section \ref{sec:fairness} we start training the models from the following set. The first model is fixed as Random Forest 1. The second model is selected from the set of following models based on their performance in the cross-validation phase. Specifically, the second model is the first one from the set whose mean accuracy score (computed via cross-validation) falls within 0.05 of the mean accuracy score of Random Forest 1.

\textbf{Classification tasks}:
\begin{itemize}
    \item \textbf{Random Forest 1}: A Random Forest classifier with 100 trees.
    \item \textbf{Random Forest 2}: A Random Forest classifier with 200 trees and a maximum depth of 10.
    \item \textbf{Gradient Boosting 1}: A Gradient Boosting classifier with 100 boosting stages.
    \item \textbf{Gradient Boosting 2}: A Gradient Boosting classifier with 200 boosting stages and a learning rate of 0.05.
    \item \textbf{KNN 1}: A K-Nearest Neighbors classifier using the default number of neighbors.
    \item \textbf{KNN 2}: A K-Nearest Neighbors classifier using 3 neighbors.
    \item \textbf{Decision Tree 1}: A Decision Tree classifier with no limit on maximum depth.
    \item \textbf{Decision Tree 2}: A Decision Tree classifier with a maximum depth of 10.
    \item \textbf{Logistic Regression}: A Logistic Regression classifier.
    \item \textbf{Ridge Classifier}: A Ridge Classifier with default settings.
    \item \textbf{AdaBoost}: An AdaBoost classifier with default settings.
    \item \textbf{GaussianNB}: A Gaussian Naive Bayes classifier with default settings.
\end{itemize}

All classifiers were imported from the scikit-learn package. For all classifiers that accepted a random state argument(all models except KNN and GaussianNB have this parameter), we set the argument to 42 for reproducibility.

The classifiers above are the ones that satisfied the meta-rule from Section \ref{sec:experiments Reconcile and MM} and therefore were added to the set $\M$.

\paragraph{Regression tasks:}
\begin{itemize}
    \item \textbf{Random Forest 1}: A Random Forest regression with 100 trees.
    \item \textbf{Random Forest 2}: A Random Forest regression with 200 trees and a maximum depth of 10.
    \item \textbf{Gradient Boosting 1}: A Gradient Boosting regression with 100 boosting stages.
    \item \textbf{Gradient Boosting 2}: A Gradient Boosting regression with 200 boosting stages and a learning rate of 0.05.
    \item \textbf{KNN 1}: A K-Nearest Neighbors regression using the default number of neighbors.
    \item \textbf{KNN 2}: A K-Nearest Neighbors regression using 3 neighbors.
    \item \textbf{Decision Tree 1}: A Decision Tree regression with no limit on maximum depth.
    \item \textbf{Decision Tree 2}: A Decision Tree regression with a maximum depth of 10.
    \item \textbf{Linear Regression}: A Linear Regression with default parameters.
    \item \textbf{AdaBoost}: An AdaBoost regression with default settings.
\end{itemize}

All models were imported from the scikit-learn package. For all regressors that accepted a random state argument(all models except KNN have this parameter), we set the argument to 42 for reproducibility.
The models above are the ones that satisfied the meta-rule from Section \ref{sec:experiments Reconcile and MM} and therefore were added to the set $\M$.

\subsubsection{Causal Estimators}
Below is a short description for each of the meta-learners we used in our experiments all of which were imported from CausalML Python
package \citep{chen2020causalml}. For more information, refer to \citep{kunzel2019metalearners,nie2021quasi}.

\textbf{R-learner.} R-learner uses the cross-validation out-of-fold estimates of outcomes \(\hat{m}^{(-i)}(x_i)\) and propensity scores \(\hat{e}^{(-i)}(x_i)\). It consists of two stages as follows:

\textbf{Stage 1}

Fit \(\hat{m}(x)\) and \(\hat{e}(x)\) with machine learning models using cross-validation.

\textbf{Stage 2}

Estimate treatment effects by minimizing the R-loss, \(\hat{L}_n(\tau(x))\):

\[
\hat{L}_n(\tau(x)) = \frac{1}{n} \sum_{i=1}^{n} \left( \left( Y_i - \hat{m}^{(-i)}(X_i) \right) - \left( W_i - \hat{e}^{(-i)}(X_i) \right) \tau(X_i) \right)^2
\]

where \(\hat{e}^{(-i)}(X_i)\), etc., denote the out-of-fold held-out predictions made without using the \(i\)-th training sample.

\textbf{T-learner.}
T-learner consists of two stages as follows:

\textbf{Stage 1}

Estimate the average outcomes \(\mu_0(x)\) and \(\mu_1(x)\):

\[
\mu_0(x) = \mathbb{E}[Y(0) \mid X = x]
\]
\[
\mu_1(x) = \mathbb{E}[Y(1) \mid X = x]
\]

using machine learning models.

\textbf{Stage 2}

Define the CATE estimate as:

\[
\hat{\tau}(x) = \hat{\mu}_1(x) - \hat{\mu}_0(x)
\]

\textbf{S-learner.}
S-learner estimates the treatment effect using a single machine learning model as follows:

\textbf{Stage 1}
Estimate the average outcomes \(\mu(x)\) with covariates $X$ and an indicator variable for treatment \(T\):
\[
\mu(x, w) = \mathbb{E}[Y \mid X = x, T = t]
\]

using a machine learning model.

\textbf{Stage 2}
Define the CATE estimate as:
\[
\hat{\tau}(x) = \hat{\mu}(x, T = 1) - \hat{\mu}(x, T = 0)
\]

\textbf{X-learner.} X-learner is an extension of the T-learner and consists of three stages as follows:

\textbf{Stage 1}

Estimate the average outcomes \(\mu_0(x)\) and \(\mu_1(x)\):
\[
\mu_0(x) = \mathbb{E}[Y(0) \mid X = x]
\]
\[
\mu_1(x) = \mathbb{E}[Y(1) \mid X = x]
\]

using machine learning models.

\textbf{Stage 2}

Impute the user-level treatment effects, \(D^1_i\) and \(D^0_j\), for user \(i\) in the treatment group based on \(\mu_0(x)\), and for user \(j\) in the control group based on \(\mu_1(x)\):

\[
D^1_i = Y^1_i - \hat{\mu}_0(X^1_i)
\]
\[
D^0_i = \hat{\mu}_1(X^0_i) - Y^0_i
\]

Then estimate \(\tau_1(x) = \mathbb{E}[D^1 \mid X = x]\) and \(\tau_0(x) = \mathbb{E}[D^0 \mid X = x]\) using machine learning models.

\textbf{Stage 3}

Define the CATE estimate by a weighted average of \(\tau_1(x)\) and \(\tau_0(x)\):

\[
\tau(x) = g(x) \tau_0(x) + (1 - g(x)) \tau_1(x)
\]

where \( g \in [0,1] \). We can use propensity scores for \( g(x) \).

\subsection{Results}\label{appendix:results}
In this section, we go through more results and analysis from our experiments.
\begin{figure}
\centering
\includegraphics[width=1\columnwidth]{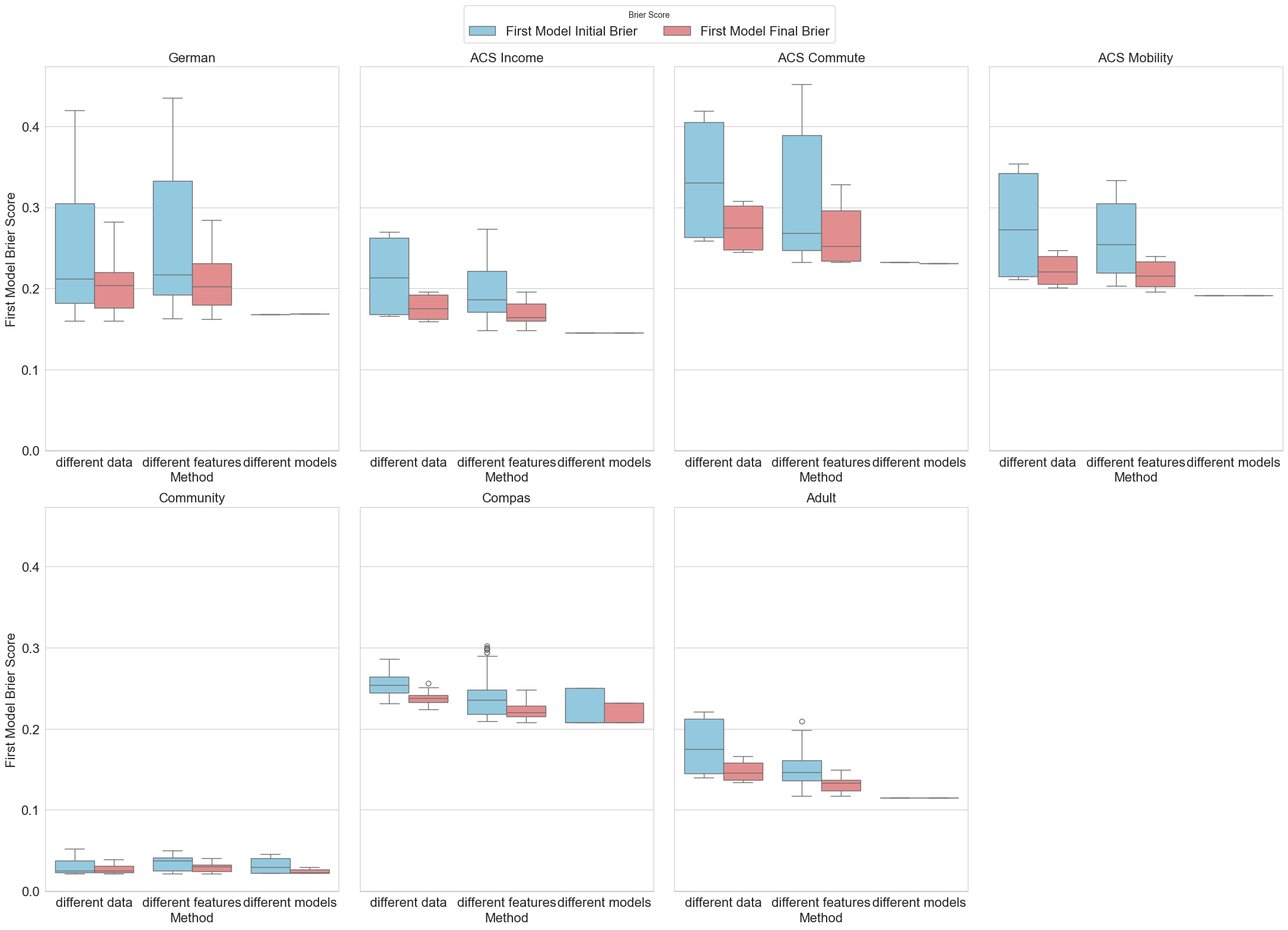} 
\caption{$f_1$ Brier Score for predictions across different datasets before and after Reconcile. Model construction follows methodologies described in Section~\ref{sec:building models}. Model specifics can be found in Section~\ref{sec:experiments Reconcile results}. $f_2$ scores follow a similar pattern and are omitted for clarity of the plot.}
\label{fig:brier_boxplot}
\end{figure}
\begin{figure}[h]
\centering
\includegraphics[width=1\columnwidth]{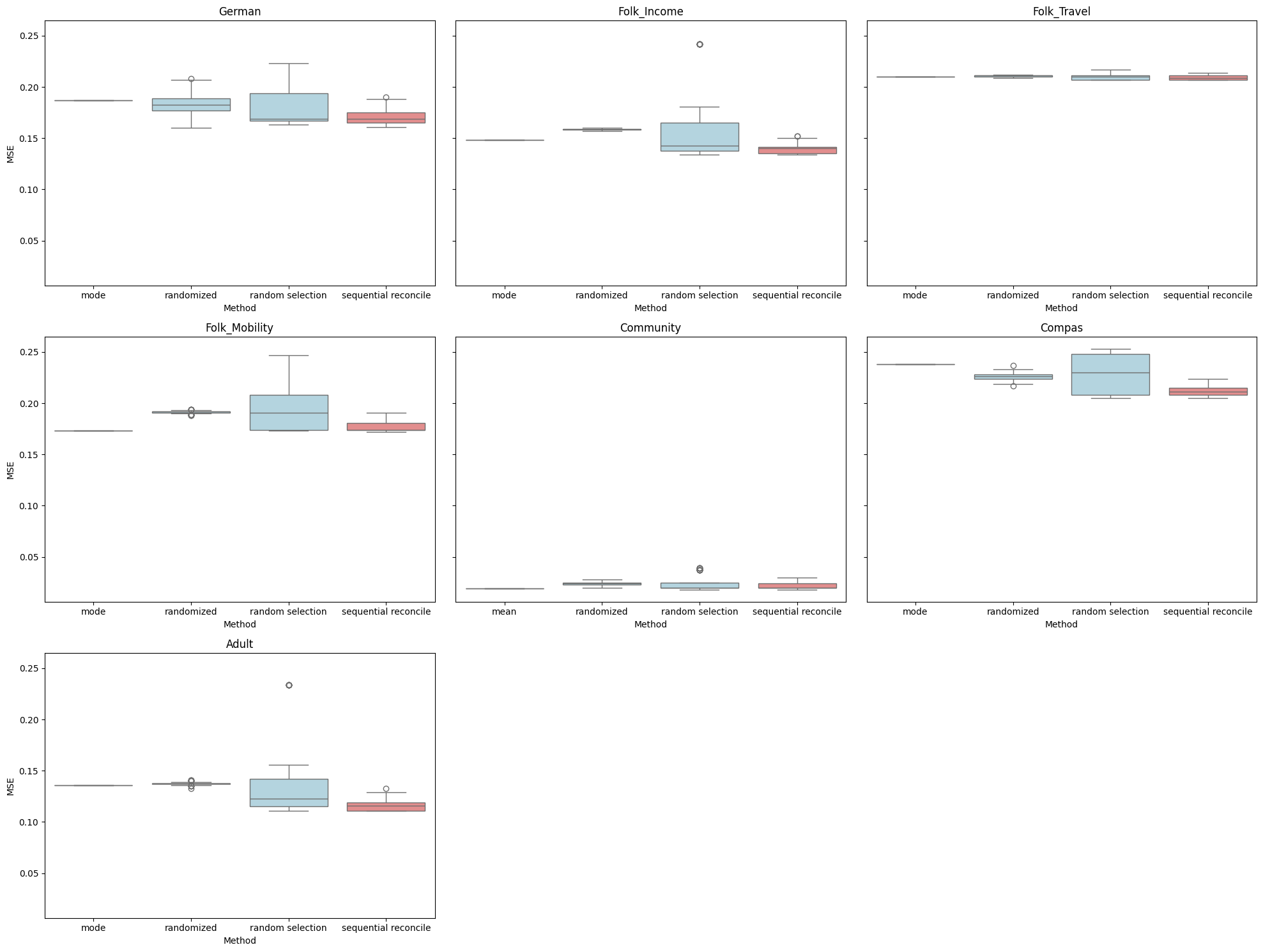} 
\caption{Boxplots comparing MSE values across different aggregation methods for each dataset (refer to Section \ref{sec:experiments Reconcile and MM} for experiment details). The Sequential Reconcile method (in red) generally shows lower variability and improved performance compared to other methods.}
\label{fig:aggregate_boxplot}
\end{figure}
\subsubsection{Robustness of Sequential Reconcile}
The experiment outlined in Section \ref{sec:experiment:stability} was repeated 20 times. For each run, we computed the differences in the MSE values of the aggregated predictors, illustrated in Figure \ref{fig:stability_mean}. We analyze the difference in MSE rather than individual values across experiments because the initial models in the set in each experiment may vary in performance. Directly comparing absolute values across experiments would not be meaningful, as differences in initial model quality could introduce variability unrelated to the effect being studied.

\begin{figure}[h]
\centering
\includegraphics[width=1\columnwidth]{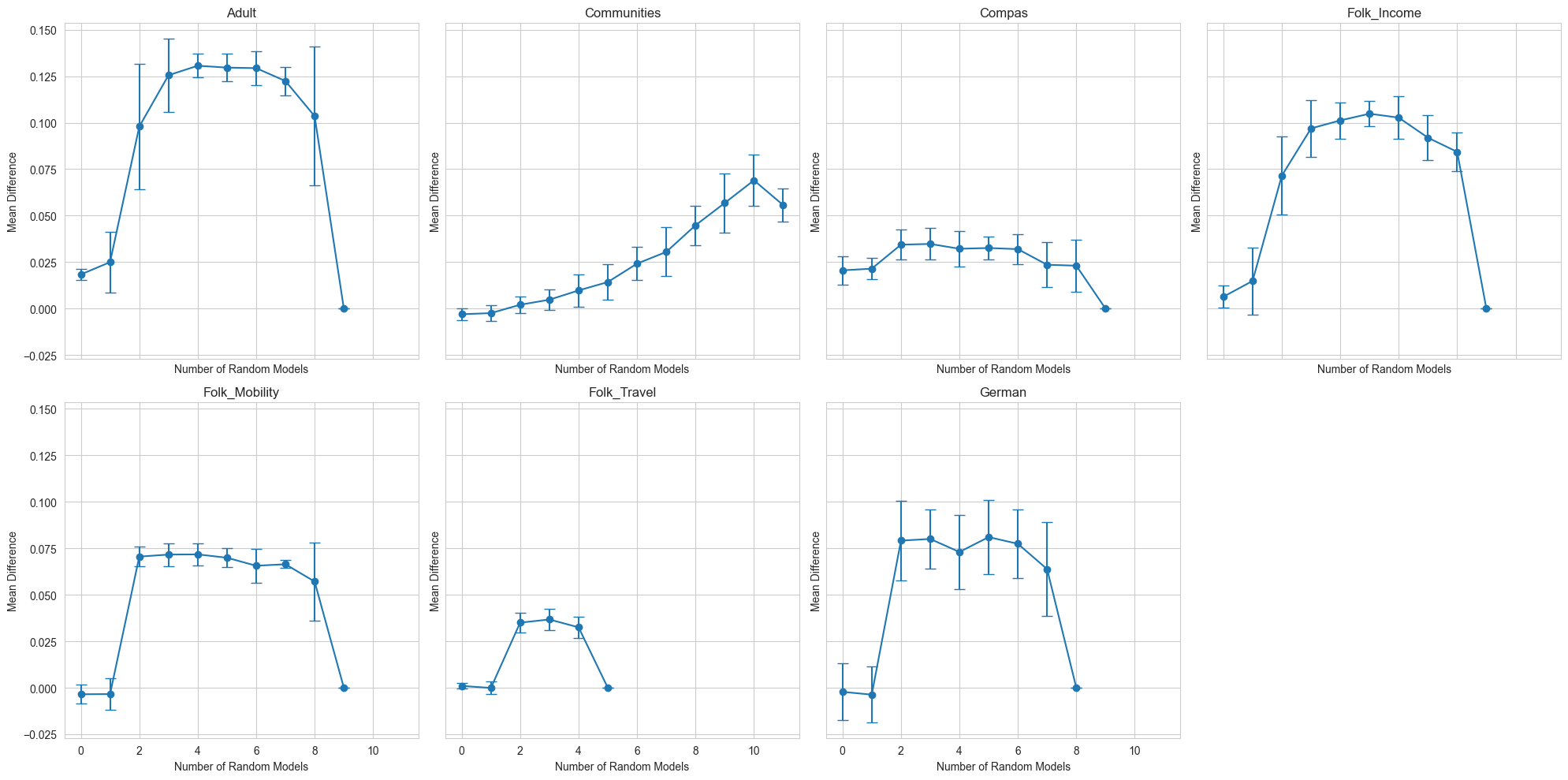} 
\caption{Mean differences in MSE between Mode aggregation (or Mean aggregation for regression in the case of the Communities dataset) and Sequential Reconcile across datasets as the number of random models increases ($MSE_{mode} - MSE_{reconcile}$). Error bars indicate variability over 20 repetitions.}
\label{fig:stability_mean}
\end{figure}

\subsubsection{Using Reconcile to Reduce Disagreements Within a Class $\M$} \label{appendix:results:within-class}
In Section \ref{sec:experiments solution}, four metrics were introduced to quantify model multiplicity within the set $\M$, though results were only discussed for one of these metrics; ambiguity. We now present the impact of applying Reconcile using the four methods described in Section \ref{sec:experiments solution} on the other three metrics.
Table \ref{tab:exp 3 Tukey discrepancy} shows the results of Tukey's HSD test on discrepancy within the sets created with each method. All that was discussed for ambiguity, holds true for discrepancy. We can see significant differences between $\M$ and the other four sets created by applying Reconcile. 
\begin{table}[h]
\centering
\begin{tabular}{|c|c|c|c|c|c|}
\hline
Set 1 & Set 2 & Mean Diff & P-adj & Upper & Lower \\ \hline
$\M$   &$\M'_a$ & $0.081$ & $0.0$ & $0.088$ & $0.073$ \\ \hline
$\M$   &$\M'_b$ & $0.106$ & $0.0$ & $0.114$ & $0.098$ \\ \hline
$\M$   &$\M'_c$ & $0.120$ & $0.0$ & $0.128$ & $0.112$ \\ \hline
$\M$   &$\M'_d$ & $0.149$ & $0.0$ & $0.156$ & $0.141$ \\ \hline
$\M'_a$   &$\M'_b$ & $0.025$ & $0.0$ & $0.033$ & $0.017$ \\ \hline
$\M'_a$   &$\M'_c$ & $0.040$ & $0.0$ & $0.048$ & $0.032$ \\ \hline
$\M'_a$   &$\M'_d$ & $0.068$ & $0.0$ & $0.076$ & $0.060$ \\ \hline
$\M'_b$   &$\M'_c$ & $0.015$ & $0.0$ & $0.022$ & $0.007$ \\ \hline
$\M'_b$   &$\M'_d$ & $0.043$ & $0.0$ & $0.050$ & $0.035$ \\ \hline
$\M'_c$   &$\M'_d$ & $0.029$ & $0.0$ & $0.036$ & $0.020$ \\ \hline
\end{tabular}
\caption{Multiple Comparison of Mean values for discrepancy across all studies for different sets - Tukey HSD, FWER=0.05}
\label{tab:exp 3 Tukey discrepancy}
\end{table}

Figures \ref{fig:disagreement} and \ref{fig:variance} visualize the statistical summaries of disagreement values and variance in predictions across sets created with different methods. 
Examining Figure \ref{fig:disagreement}, we observe that the minimum disagreement probability mass between model pairs is zero for all methods except (c) and (d). Since we selected models that perform similarly well (as defined by the meta-rule in Section~\ref{sec:experiments Reconcile and MM}), it is highly plausible that at least one pair of models will not have significant disagreement. For methods (a) and (b), where all possible model pairs are considered for Reconcile, there is a high probability that at least one model remains unchanged throughout the process. Because the same predictor appears multiple times in the set, the disagreement between identical models is always zero, ensuring that at least some models in the set exhibit no disagreement. However, method (c) selects $|\M|$ models at random, and in method (d), the order of model selection is randomized. Our hypothesis is that this random selection process increases the likelihood of scenarios where the minimum disagreement between models is nonzero, explaining the observed differences in minimum disagreement distributions.The same hypothesis applies to the maximum variance in predictions (see Figure~\ref{fig:variance}).

Aside from these two anomalies, the patterns observed in both disagreement values (Figure \ref{fig:disagreement}) and variance in predictions (Figure \ref{fig:variance}) align with the results of the ambiguity and discrepancy metrics. Specifically, we observe that applying Reconcile using any of the methods (a), (b), (c), or (d) consistently reduces the range of mean and variance values for these metrics.

\begin{figure}[h]
\centering
\includegraphics[width=15cm]{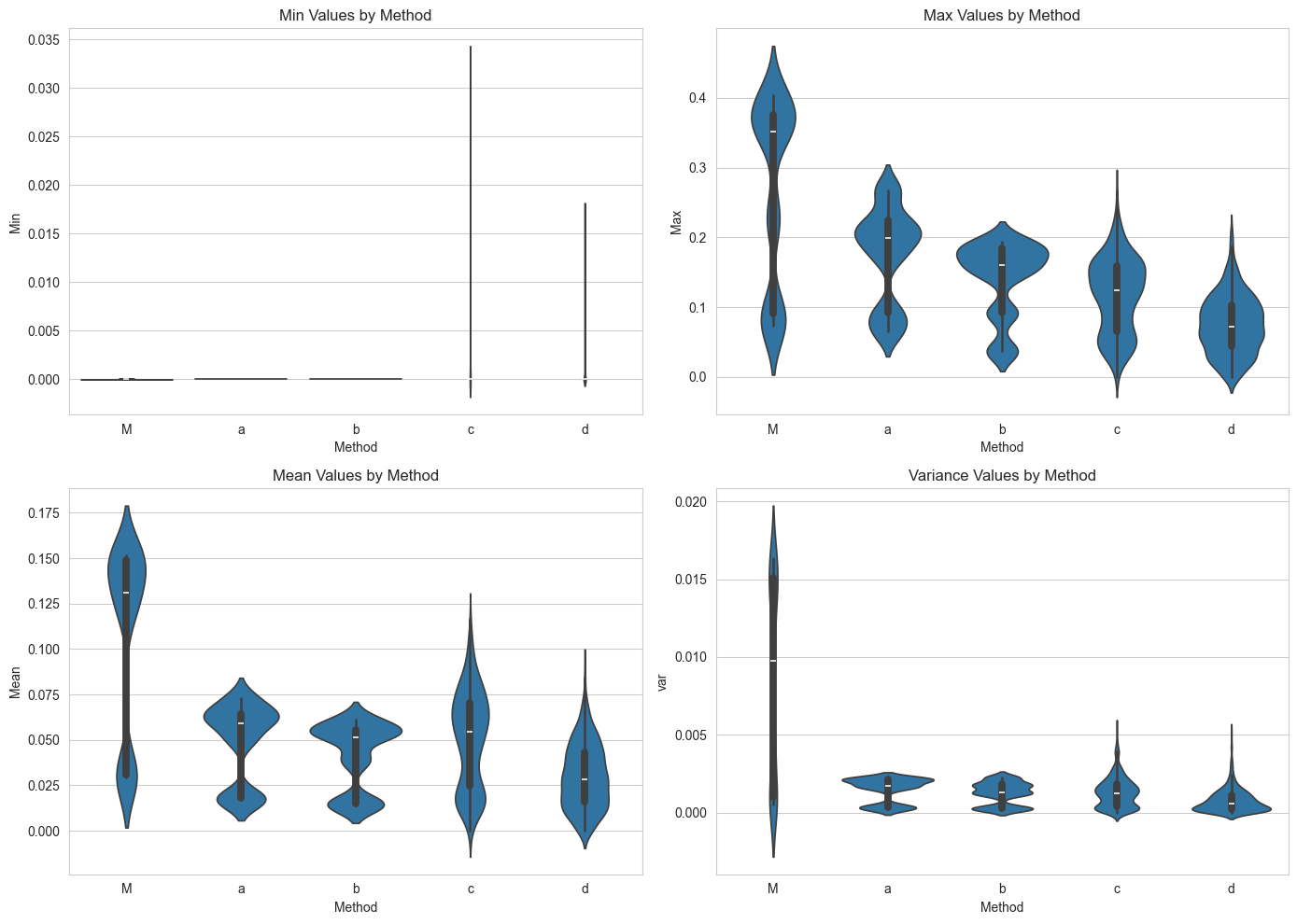}
\caption{For all 2-combinations of models \( f_1 \) and \( f_2 \) in $\M$, \( \mu(U_{0.2}(f_1, f_2))) \)) was calculated. The four plots show the distribution of the statistics(min, max, mean and variance) for different runs.} \label{fig:disagreement}
\end{figure}

\begin{figure}[h]
\centering
\includegraphics[width=16cm]{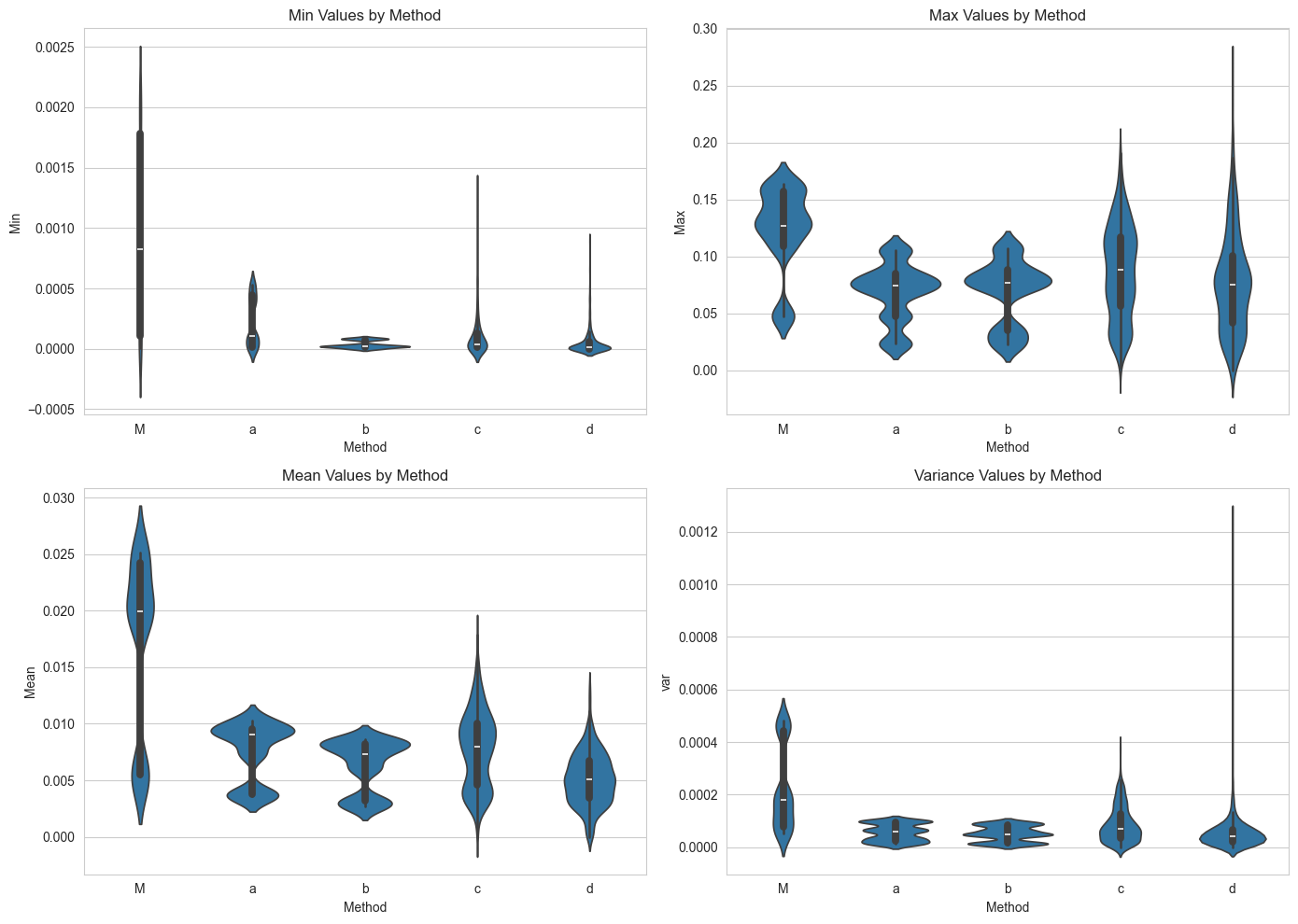}
\caption{For each instance $x_i$ in $D$, we calculate the variance(\(\sigma^2(x_i)\)) of the predictions made by all models in the set $\M$. The four plots show the distribution of the statistics(min, max, mean, and variance) of these variances across the entire set $D$ for different runs.} \label{fig:variance}
\end{figure}

\subsubsection{CATE Experiments}
Figure \ref{fig:CATE_disagreement_national} shows disagreement levels on the national study dataset before and after Reconcile. Figures \ref{fig:CATE_brier_twin} and \ref{fig:CATE_brier_national} present the Brier score of the first estimator before and after applying Reconcile on the Twins dataset and the National Study dataset. The second estimator exhibited a similar trend and was omitted for clarity. In all cases, we observe an improvement in the Brier score, even if marginal.

\begin{figure}
    \centering
    \includegraphics[width=14cm]{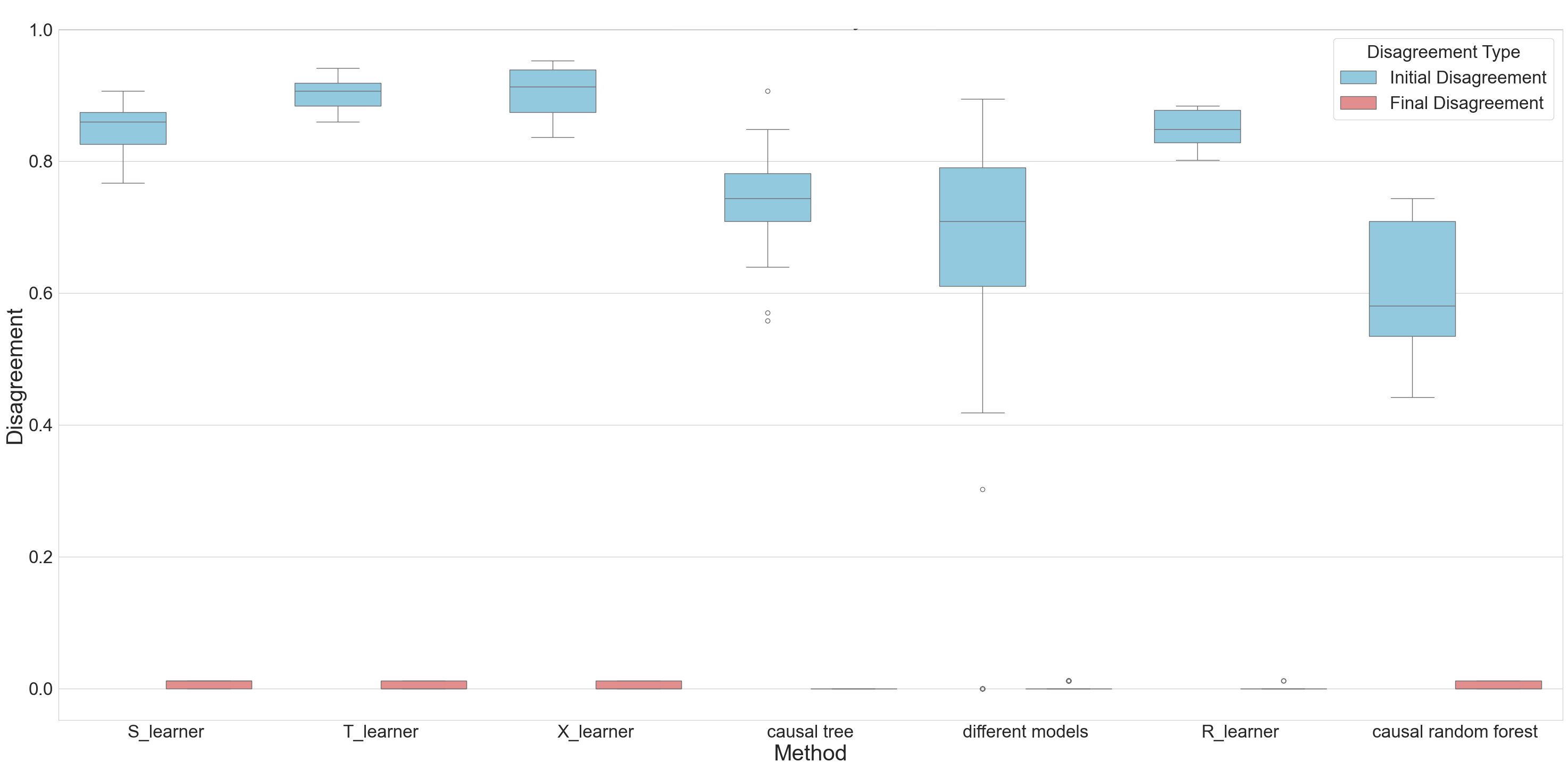}
    \caption{Disagreements between CATE estimates before and after running Reconcile on National Study dataset.}
    \label{fig:CATE_disagreement_national}
\end{figure}

\begin{figure}
    \centering
    \includegraphics[width=0.9\textwidth]{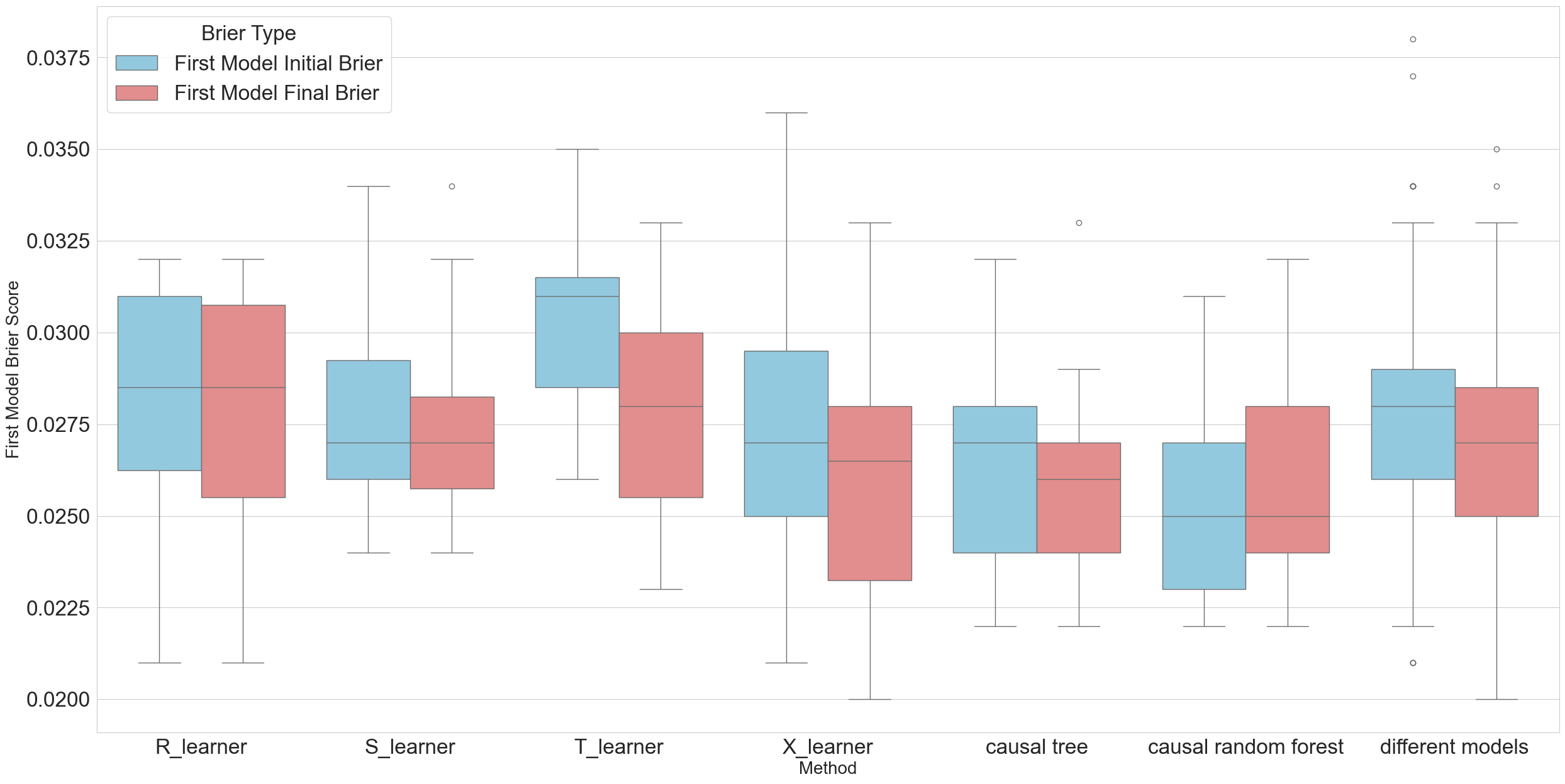}
    \caption{First estimator's Brier Score across different estimators before and after Reconcile on Twins dataset.}
    \label{fig:CATE_brier_twin}
\end{figure}

\begin{figure}
    \centering
    \includegraphics[width=0.9\textwidth]{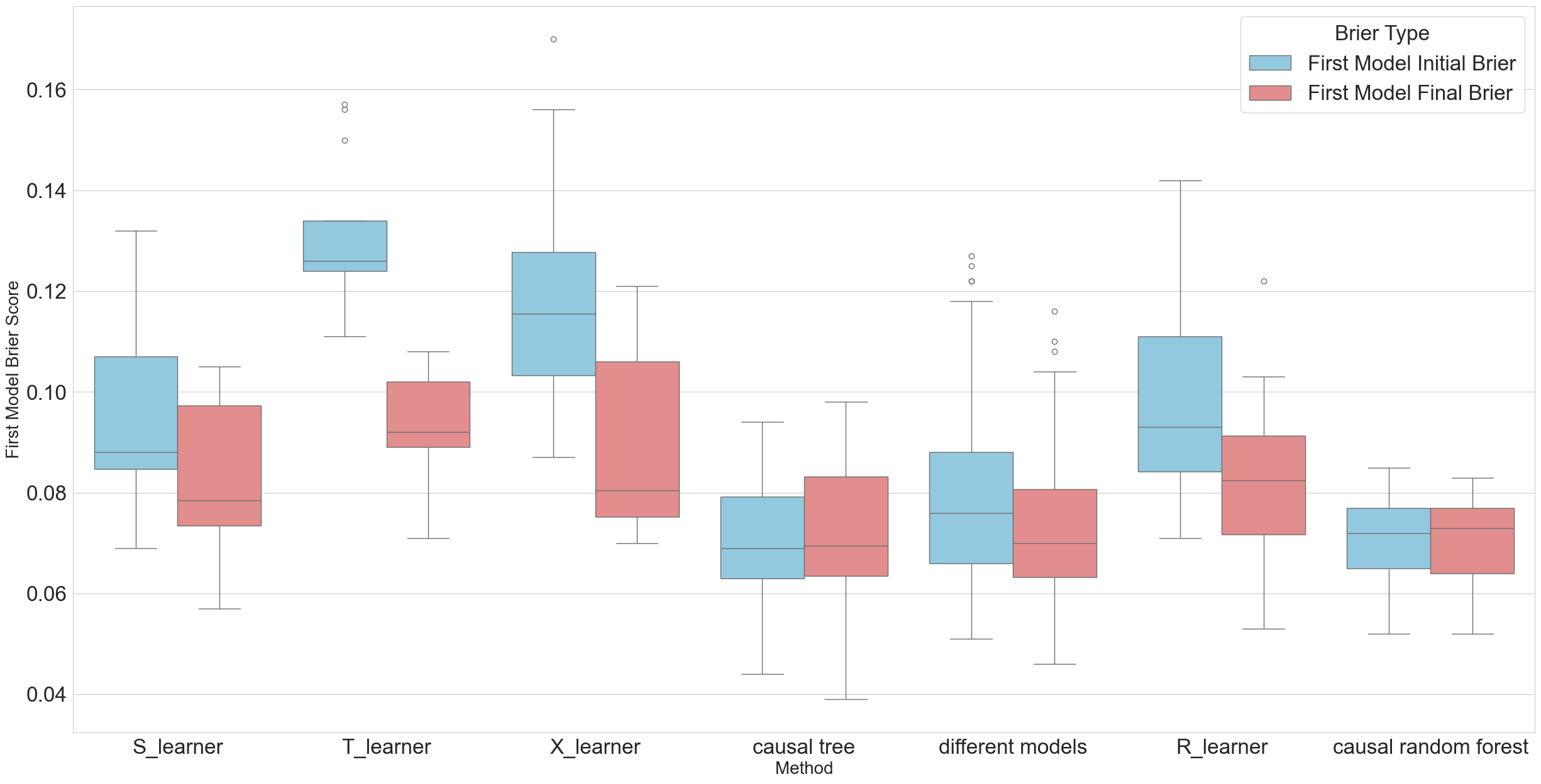}
    \caption{First estimator's Brier Score across different estimators before and after Reconcile on National Study dataset.}
    \label{fig:CATE_brier_national}
\end{figure}

\end{document}